\documentclass[conference]{IEEEtran}
\usepackage{tikz}
\usepackage{booktabs}
\usepackage{graphicx}
\usepackage{textcomp}
\usepackage{xcolor}
\usepackage{algorithm2e}
\usepackage{listings}
\usepackage{tabularx}
\usepackage{flushend}
\usepackage{soul}
\usepackage{wasysym}
\usepackage{multirow}
\usepackage{url}
\usepackage[T1]{fontenc}
\usepackage{microtype}
\usepackage{nowidow}
\usepackage{subcaption}
\usepackage{xspace}
\usepackage{enumitem}
\usepackage{amsthm}
\usepackage{algpseudocode}
\usepackage{pseudo}

\newcommand*{\rp}{\textit{RP}\xspace}
\newcommand*{\hsk}{\textit{HSK}\xspace}

\newtheorem{theorem}{Claim}

\newcommand{\ie}{\textit{i}.\textit{e}.}

\newcommand{\A}{$\mathcal{A}$\xspace}

\begin{document}

\title{A Security and Usability Analysis of \\ Local Attacks Against FIDO2}

\author{
    Tarun Kumar Yadav,
    Kent Seamons\\
    Brigham Young University\\
    {tarun141@byu.edu, seamons@cs.byu.edu}\\
    % \textbf{In submission: Do not Redistribute}
}

% \author{Anonymous}

\maketitle

\begin{abstract}
The FIDO2 protocol aims to strengthen or replace password authentication using public-key cryptography. 
FIDO2 has primarily focused on defending against attacks from afar by remote attackers that compromise a password or attempt to phish the user. In this paper, we explore threats from local attacks on FIDO2 that have received less attention---a browser extension compromise and attackers gaining physical access to an HSK.
Our systematic analysis of current implementations of FIDO2 reveals four underlying flaws, and we demonstrate the feasibility of seven attacks that exploit those flaws.
The flaws include (1) Lack of confidentiality/integrity of FIDO2 messages accessible to browser extensions, (2) Broken clone detection algorithm, (3) Potential for user misunderstanding from social engineering and notification/error messages, and (4) Cookie life cycle.
We build malicious browser extensions and demonstrate the attacks on ten popular web servers that use FIDO2. We also show that many browser extensions have sufficient permissions to conduct the attacks if they were compromised. A static and dynamic analysis of current browser extensions finds no evidence of the attacks in the wild. We conducted two user studies confirming that participants do not detect the attacks with current error messages, email notifications, and UX responses to the attacks.
We provide an improved clone detection algorithm and recommendations for relying parties that prevent some of the attacks. 
\end{abstract}

% \begin{CCSXML}
% <ccs2012>
% <concept>
% <concept_id>10002978.10002991.10002992</concept_id>
% <concept_desc>Security and privacy~Authentication</concept_desc>
% <concept_significance>500</concept_significance>
% </concept>
% </ccs2012>
% \end{CCSXML}

% \ccsdesc[500]{Security and privacy~Authentication}

% \settopmatter{printfolios=true}

% \input{abstract}

\section{Introduction}

Two-factor authentication (2FA) defends against account compromise due to stolen passwords and phishing attacks. 
The current state-of-the-art for 2FA on the Web is FIDO2, the FIDO (Fast Identity Online) Alliance, and the W3C's newest set of specifications supporting 2FA, multi-factor authentication (MFA), and passwordless authentication. FIDO2 provides a standard web services API (WebAuthn) on client machines to authenticate users using public-key cryptography. The API is available in all popular browsers and seeing adoption by major service providers, including Facebook, GitHub, and Gmail.

FIDO2 supports a variety of client-side authenticators, including hardware security keys ({\hsk}s) and built-in platform authenticators such as biometrics. Although this paper refers to {\hsk}s, the ideas apply to all FIDO2 authenticators.

The FIDO2 specification focuses on remote attackers and thus their threat model assumes a trusted client (browser and browser extensions). However, malicious browser extensions have a long-standing history of stealing user data, including passwords, financial information, and browsing history ~\cite{ositcom, awakesecurity, catonetworks}. With the adoption of FIDO2 protocol, it is crucial to recognize the potential risks posed by these malicious extensions. Examples of password theft by malicious extensions such as "Web Security" and "Stylish" underscore the pressing need to study these attacks for new emerging protocols such as FIDO2 and develop countermeasures to mitigate the risks. As attackers constantly evolve their tactics, it is important to proactively research potential attack vectors to stay ahead of the curve and ensure the effectiveness of FIDO2's security measures. The threat posed by these extensions cannot be overlooked and warrants continued research, mitigation strategies, and vigilance against new and emerging threats.

% The FIDO2 specification focuses on remote attackers and thus assumes a trusted client 
% However, browser extensions have access to FIDO2 messages and are vulnerable to compromise.
% We show that the threat to FIDO2 authentication from compromised browser extensions is real. 
% Furthermore, we determine that many extensions have sufficient permissions to attack FIDO2. Our interest in local attacks against FIDO2 began when we observed that FIDO2 messages are not integrity-protected. We hypothesized that a man-in-the-middle (MITM) attacker could modify a FIDO2 registration response message by replacing the user's public key with the attacker's public key. We created a malicious browser extension to attack our own account to confirm that the attack was feasible. This initial experiment launched our in-depth study of browser extension attacks on FIDO2. 

Our research systematically analyzed local attacks against FIDO2 from malicious browser extensions. 
% We conducted a systematic analysis of local attacks against FIDO2 that includes malicious browser extensions. 
During our analysis, we also discovered a weakness in the clone detection algorithm that combats attackers that gain physical access to {\hsk}s. 
As a result, we identified four fundamental flaws that attackers can exploit: (1) Lack of confidentiality/integrity of FIDO2 messages accessible to browser extensions, (2) Broken clone detection algorithm, (3) Potential for user misunderstanding from social engineering and notification/error messages, and (4) Cookie life cycle.

We describe seven attacks that exploit these flaws and implement the attacks to demonstrate their feasibility. 
Four of the attacks have not been described previously. Three of them have been described earlier as theoretical attacks, but we implement them and show they are feasible (see Table~\ref{flaw-table}).

We were surprised to discover two attacks (3 and 4) where passwordless authentication presents more risk than passwords alone in the presence of a compromised extension.
After performing the attack, an attacker can log in to an account the victim never logs into from the vulnerable browser. These attacks are significant because users may log into only a low-value account from an untrusted computer and not anticipate this putting their high-value accounts at risk if the computer is compromised. This risk did not exist in the world of passwords when a user had a different strong password on their high-value account. Surprisingly, passwordless authentication using public-key cryptography opened up a new risk.

%We design our defense v-FIDO2 considering a future-looking stronger threat model that assumes an untrusted client. We recognize that making design changes after widespread adoption can be challenging, and a proactive approach is necessary to prevent future security breaches. Sharing a stronger v-FIDO2 defense may inform and influence future solutions. Our aim is only to improve the security of FIDO2 under an untrusted threat model. Other attacks are outside our threat model and are being considered by other researchers~\cite{eskandarian2019fidelius}.

We also determine that 47\% of browser extensions on the Chrome Web Store have sufficient permissions to execute most of the attacks.
Finally, we present an improved clone detection algorithm that an \hsk's firmware can implement and make recommendations for web servers to improve the security of FIDO2 implementations. 

\begin{figure}[htbp]
\centerline{\includegraphics[width=0.5\textwidth]{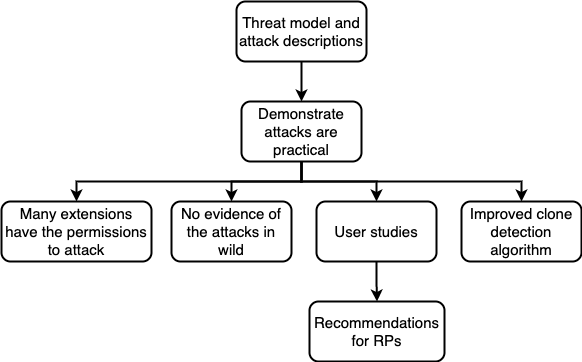}}
\caption{Paper overview}
\label{paper_overview}
\end{figure}

% We also design and analyze remediations that defend against the attacks. We
% Finally, we present recommendations for web servers to improve the security of FIDO2 implementations.

%While analyzing attacks from a compromised browser extension against FIDO2 implementation (\#1-\#7), we also discovered an attack (\#8) against the device cloning detection algorithm presented in the FIDO2 specification. We present a new device cloning detection algorithm to prevent the attack. 

% \tarun{Lyastani et al. shows that  ~\cite{lyastani2020fido2} 61\% of users were mainly afraid that someone else could gain access to their accounts with a lost or stolen security key. Thus, protocol improvements to protect against the attackers with physical access to the hardware token could increase users' confidence and, in turn, potential adoption.}
% \ks{The fear is a stolen device the attacker uses, so the clone detection seems more like a special case. Not sure this is a strong argument. I also hesitate to make claims about adoption, a difficult issue. I think move this statement to discussion if we keep it.}

An overview of the paper is depicted in Figure \ref{paper_overview}. The contributions of this paper are the following:
\begin{enumerate}
% [topsep=1pt,noitemsep,leftmargin=1.5em]
 
\item 

Systematization of attacks by two local adversaries on FIDO2 security: attacks from a malicious browser extension and an attacker gaining physical access to an \hsk. Our systematization reveals seven practical attacks or weaknesses.
\item Demonstrate the feasibility of the attacks: 
\begin{itemize}
% [topsep=1pt,noitemsep,leftmargin=1.5em]
\item We prototype a malicious browser extension for Chrome and Firefox, and demonstrate the attacks on ten popular web servers that use FIDO2.
\item We analyze current Google Chrome extensions to identify (1) how many have sufficient permissions to execute the attacks if they were compromised and (2) the scale of such attacks based on the number of users for the extensions. Our results show that 105,381 out of 211,026 extensions have sufficient permissions to compromise a WebAuthn client and execute the attacks. Furthermore, 404 of these extensions have more than one million users each.
\item We confirm that the current clone detection algorithm is vulnerable to a stealthy device cloning attack
\item We perform a static and dynamic analysis of current browser extensions and find no evidence of these attacks in the wild.
\end{itemize}
\item Through two user studies (n=80 and n=20), we confirm that current error messages, email notifications, or changes in the UX caused by these attacks were insufficient for participants in the user study to detect them.
\item Present an improved clone detection algorithm.
\item Provide recommendations for web servers and browsers that could help users detect some of the attacks.

% \begin{itemize}
% % [topsep=1pt,noitemsep,leftmargin=1.5em]

% \item The design of an improved clone detection algorithm that defends against stealthy device cloning attacks.
% The improved algorithm can be easily integrated into {\hsk}s and FIDO2 libraries, independent from other defenses in the paper. 
% % independent from other defenses presented in the paper
% % incorporated into the FIDO2 specification deployed by security key vendors.}
% % \ks{Not sure we need this sentence. What is the main reason for this addition?}

% \item Specification of v-FIDO2 designed to protect against the attacks: 
% We proposed v-FIDO2 architecture that leverages the Trusted Execution Environment (TEE) in our design to create an authenticated communication channel between the \hsk and the web application. It prevents modification of the FIDO2 messages by a compromised browser extension or anywhere along the communication path. 
% We conduct a threat analysis of the v-FIDO2 design under a more rigorous threat model that includes a compromised operating system and develop a proof-of-concept prototype.
% This defense protects against all eight attacks from not only a malicious browser extension but also a compromised browser or operating system.

% \item Recommendations for web servers to help prevent and detect some attacks.

% % \item Discussion of an improvement to the WebAuthn API browser architecture to prevent malicious/vulnerable extensions or webpages from modifying the FIDO2 registration parameters. This defense protects against some of the attacks and is easy to deploy.

% \end{itemize}

\end{enumerate}

\section{FIDO2 Background}
\label{background}
The FIDO2 protocol is a secure authentication method that utilizes public-key cryptography for user authentication to web services. In this protocol, users register their public key with the web service and prove ownership of the corresponding private key by signing challenges presented by the service.

The FIDO2 protocol involves three main entities:

The \textit{Relying Party (\rp)}: This entity represents the web application, such as "facebook.com," that supports FIDO2 authentication. The \rp communicates with the authenticator through the WebAuthn client. It can choose to offer FIDO2 as a method for two-factor authentication (2FA) or passwordless authentication.

The \textit{Authenticator}: This entity is a device that securely stores the user's private key. Users authorize an authenticator to generate a login credential to the \rp by providing some form of input, such as a PIN or a button press. FIDO2 supports various types of authenticators, including built-in fingerprint readers and external (remote) authenticators like hardware security keys (\hsk).

The \textit{WebAuthn client}: Typically present in a web browser, the WebAuthn client acts as a mediator between the authenticator and the \rp. It relays information and commands between the two entities. To prevent phishing attacks, the WebAuthn client reports the origin (URL) of the \rp to the authenticator, ensuring that the authentication process is bound to the correct website.

The FIDO2 protocol consists of two main components: the Web Authentication (WebAuthn) browser API and the client-to-authenticator protocol (CTAP). The WebAuthn API in the client provides an interface for the \rp to interact with the authenticator, while the CTAP protocol enables secure communication between the WebAuthn client and external or roaming authenticators using Bluetooth, USB, or NFC.

\subsection{Registration and authentication}
The registration and authentication processes in FIDO2 involve the following steps:

Registration: Users initiate the registration process by clicking a "register/login" button. The \rp sends a registration request to the WebAuthn client, including a challenge, user information, and \rp information. The WebAuthn client forwards this information to the \hsk, along with the \rp's origin and the request type. The \hsk prompts the user for consent, generates a new asymmetric key pair, and sends the registration data (credential ID, public key, RP ID hash, counter, and attestation signature) as an attestation object to the WebAuthn client. The WebAuthn client then forwards this data to the \rp, which verifies the signatures and critical parts of the response. Upon successful verification, the \hsk is registered to the user's account.

Authentication: The authentication process is similar to registration, with a few differences. Authentication does not require user information, and instead of attestation, the \hsk performs assertion by signing the response with the private key corresponding to the credential ID.

These mechanisms ensure secure and reliable authentication using the FIDO2 protocol, providing enhanced protection against various attacks and unauthorized access attempts.

\begin{figure*}[ht]
\centerline{\includegraphics[width=0.8\textwidth]{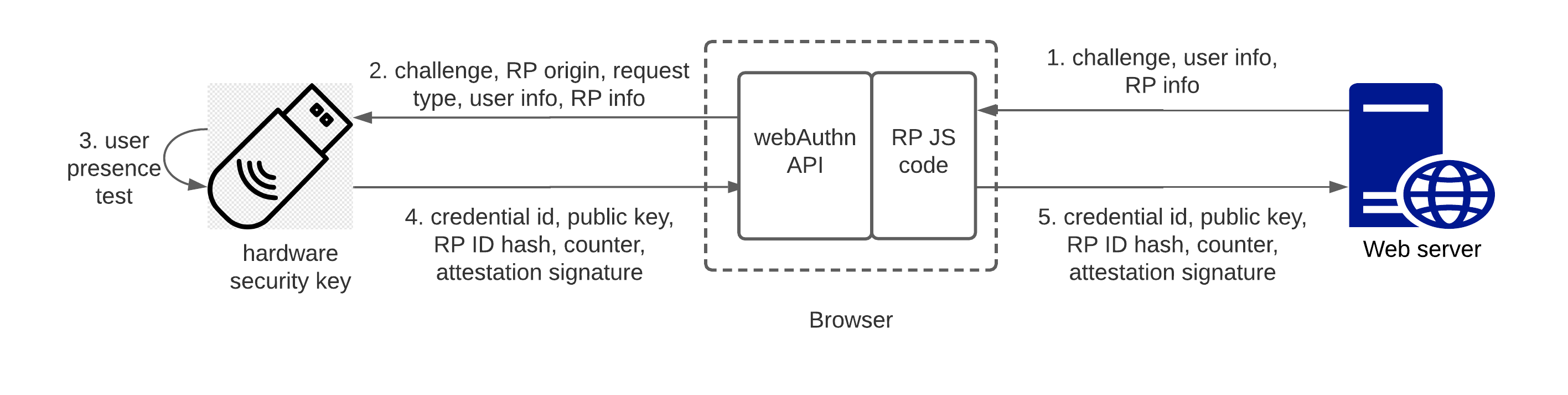}}
\caption{FIDO2 Registration}
\label{FIDO2_registration}
\end{figure*}

\subsection{Clone detection}
\label{clone_detect}

By design, an \hsk never releases the user's private key and other sensitive data. In theory, an attacker must steal an \hsk to impersonate a user. Ninjalab recently demonstrated a successful cloning of a Google Titan Security Key using side-channel attack.
The NinjaLab attack requires access to the device, 10 hours, \$12,000 in equipment, and specialized expertise \cite{lomne-titan}. 
% One instance of a clone detection attack on Google Titan through one approach cost $12$k and $10$ hrs.
However, there are various HSK vendors with different firmware and hardware; attackers will find other ideas to clone HSKs faster at a cheaper cost. Roth et al. ~\cite{roth2022airtag} demonstrate the cloning attack on Nordic nRF52832 by voltage glitching the nRF chip for firmware extraction with a low-cost setup (5€). 

% \ks{Note to Garrett - create a citation to the report at this URL \url{https://ninjalab.io/wp-content/uploads/2021/01/a_side_journey_to_titan.pdf}}

To detect cloning, the \hsk and \rp both maintain a counter. An \hsk can have account-specific counters or a global counter. During registration, the \hsk initializes the counter on the device and sends it to the \rp, as shown in Figure~\ref{FIDO2_registration}. The \hsk increments the counter and sends it to the \rp each time it authenticates. The \rp confirms the received counter is larger than the current counter and updates the counter. A cloning attack is detected if the \rp ever receives a counter lower than the current counter. 

If an attacker clones a device and impersonates the user, the \rp increments the counter each time the attacker authenticates. When the victim eventually authenticates using the original \hsk and counter, the \rp detects the attack because the counter it receives is lower than its current counter and notifies the user about it. 

% \subsection{Clone detection - OLD}
% \label{clone_detect}
% The private keys in an \hsk cannot be read externally even if an attacker gains physical access to the \hsk. However, the device itself may be cloned. To detect cloning, the \hsk and \rp both maintain an RP-specific counter that indicates the number of times the \hsk has authenticated to the \rp. 
% During registration, the \hsk creates an initial value of the counter and sends it to \rp, as shown in fig ~\ref{FIDO2_registration}. The \hsk increments the counter every time it authenticates and sends a copy of it to the \rp.
% \st{during authentication, similar to the registration process shown in fig} 
% \ks{Similar to registration process? Now that I think about it, why is a counter sent during registration? Doesn't it just start at 0 no both ends?}  

% \tarun{[The counter doesn't have to start with 0. They just need to agree on the same counter value, that's what they do during registration. The \hsk sends an initial counter value to \rp during registration and increments during authentication.]}  

% If the received counter is higher than the previously-stored counter at the \rp, the \rp allows authentication and updates the counter to the value it received. In case an attacker clones the device and successfully authenticates as the user, the \rp increments the counter, but the original \hsk has the previous counter. When the user eventually uses the original authenticator, the \rp generates an error when the counter it receives is lower than the current counter stored at the \rp.

\subsection{Attestation}
\label{attestation}
During registration, an optional attestation process allows an \rp to verify the make and model of the \hsk. 
Attestation lets the \rp allow only specific devices (e.g., all employees must use a Google Titan Key) or block models with known security flaws.

Each \hsk comes with a hard-coded private attestation key shared among a group of {\hsk}s, such as $40K$ devices, to prevent user tracking. An \hsk proves its make and model by signing part of the registration response with its attestation key. 
An \rp needs access to the \hsk public attestation keys to verify attestation signatures. The \rp can access the keys on demand at the vendor or maintain a local copy.

\section{Adversary Model and Attacks}
    \label{sec:adversary_model}

This section provides an overview of the entities in FIDO2 and how they communicate. We then introduce our adversary model and describe nine attacks.
%and their feasibility for WebAuthn API and the \hsk resp.
    \paragraph{Entities and communication}
    In FIDO2, there are three main entities: (1) Relying Party (RP): a Web service, (2) WebAuthn Client: relays communication between an \rp and an \hsk, and (3) \hsk: a user's hardware authenticator. To initiate the registration or authentication process with an \rp, the user communicates with the webAuthn client. The webAuthn client communicates with the \rp and the user's \hsk to register or authenticate the user. The user may be asked to touch the \hsk to confirm a request and prove that the user initiated the request. A more detailed explanation of the entities is provided in section~\ref{background}.
    % \ks{Fix this reference and search document for any other ??}
    
    \paragraph{Adversary model}
        Our adversary model includes two independent adversaries:
        \begin{itemize} 
            \item Adversary A1: a malicious/compromised browser extension or malicious web pages that leverage a vulnerable extension to compromise the FIDO2 webAuthn API. The adversary has access to plaintext FIDO2 communication and therefore can execute various attacks, such as MITM. 
            \item Adversary A2: An adversary that gains temporary physical access to the victim's \hsk and has cloned the device. The adversary cannot retrieve the private key or any other metadata from the \hsk, they can only clone it. 
        \end{itemize}

        % Table~\ref{attacks} shows which attacks can be performed by individual capabilities of A1. Capability C4 is the only one able to perform all six of the MITM attacks we identified. 
        
        \newcommand*{\yes}{{\tiny \CIRCLE}}
\newcommand*{\no}{{\tiny \Circle}}
\newcommand*{\mandatory}{{${+}$}}

\begin{table*}[htbp]
\caption{Relying Party Analysis}
\begin{center}
\begin{tabular}{|c|c|c|c|c|c|p{6cm}|}
\hline

& Relying & Attestation & Authentication  & Notification & Least secure & Clone detection error \\
 & Party & & before adding  & after adding & signature &  \\
 & & & additional \hsk & an \hsk & algorithm &  \\
\hline
1 & Facebook & \no & password & \no & ES256& (Page refresh) \\
2 &Github & \no & \no & email & RS256& "Security key authentication failed"\\

3 &Boxcryptor & required & password & \no & ES256 & N/A $\ddagger$ \\

4 &Dropbox & required & password & email & RS256& N/A $\ddagger$\\

5& Twitter & \no & N/A $\dagger$ & \no  & RS256 & (technical problem during testing)\\

6& Cloudflare & \no & password & \no & RS256 & "Invalid security key used. Please use a security key registered to this account." \\

7& Basecamp	& \no & \no & email & RS256 & "We couldn't verify this security key. Make sure you have registered it." \\

8& Login.gov & \no & \no & \no & RS256	 & (cloning not detected) \\

9& Shopify & \no & password & email & RS256 & "Couldn't connect to your security key. Try again." \\

10& 1Password & \no & \no & \no & ES256 & "Unable to verify your security key." \\

% copy& More table copy$^{\mathrm{a}}$& &  \\
\hline
% \multicolumn{4}{l}{${*}$ login session is very short lived}\\
\end{tabular}
\end{center}

% \vspace{1mm}
\begin{tabular}{l}

\no~none \space\space\space\space $\dagger$ Twitter supports only one \hsk on an account
\space\space\space\space
$\ddagger$ Unable to test clone detection due to required attestation\\

~~~$ES256 =$ $ECDSA$ w/ $SHA$-$256$
\space\space\space\space\space\space\space\space\space\space\space\space\space\space\space
\space\space\space\space\space\space\space\space\space\space\space\space\space\space\space\space\space
$RS256 (-256)=$ $RSASSA$-$PKCS1$-$v1\_5$ using $SHA$-$256$

\end{tabular}
\label{rp-analysis}
\end{table*}

        \paragraph{Adversary goals}
    % \tarun{Need to justify why having a short-term attack and access on attackers' machine is realistic and more dangerous. }
        
        Adversary A1's goal is to impersonate a victim (Bob) who uses an \hsk and gain unauthorized access to Bob's account. Ultimately, A1 wants to impersonate Bob from \textit{A1's device} over an \textit{extended period} without detection after executing a short-duration one-time attack from Bob's device or another device that Bob uses just once. Reliance on a single malicious code execution from Bob's device removes A1's continued dependence on an extension's malicious code to access Bob's accounts and, therefore, decreases the chances of detection.
        
        A1 cannot achieve its goal by registering OAuth tokens, stealing session cookies, or monitoring all user communication. Many websites do not use OAuth access tokens (e.g., banking), and the tokens last only for several hours to a couple of weeks. Furthermore, cookie-based login sessions expire as soon as a user logs out. Therefore, if the attacker wants long-term access, they have to steal cookies frequently, which is infeasible if the victim only logs in only once from a vulnerable browser.
        
        Adversary A2's goal is to gain unauthorized access to the victim's account using the cloned device \textit{without being detected}. In other words, their goal is to bypass FIDO2's clone detection algorithm.
        
\subsection{Attacks}
\label{sec:attacks}
This section provides detailed descriptions of attacks that Adversary A1 can execute on the webAuthn client API and Adversary A2 can execute with a cloned \hsk.
We take a holistic approach to explore a range of attacks in detail based on the four main flaws. The types of flaws include protocol errors, HCI challenges that impact user understanding, and implementation weaknesses. Table~\ref{flaw-table} shows the flaws, type of flaw, and attacks that exploit the flaw. We are the first to demonstrate the practical feasibility of these attacks. 

\subsubsection{Mis-binding attack during registration -- Attack 1}
\label{misbind}
During registration of the victim's \hsk, Adversary A1 replaces the victim's public key with the attacker's public key in the \hsk response. It also replaces the digital signatures with signatures generated using the attacker's private key. This attack causes the \rp to register the attacker's \hsk instead of the victim's \hsk. This attack was first identified by Hu et al, when investigating the security of the Universal Authentication Framework (UAF), a precursor to FIDO2\cite{hu2016security}.
    
An attacker can register either a software-based \hsk or a hardware-based \hsk. A software-based \hsk makes it easier for the attacker to automate the MITM attack but is not an option if the \rp requires \textit{attestation}. 
From our analysis, popular RPs like Facebook and GitHub do not require \textit{attestation} (see Table~\ref{rp-analysis}).
Attestation forces the attacker to use legitimate hardware \hsk, which increases the attacker's effort because the attacker must forward the requests and responses to a remote machine where they can connect the hardware \hsk. The attacker can use an \hsk with the same make and model as the user to be more stealthy. The attacker can automate the key tap on their \hsk by building additional hardware to perform the key tap without the user being present.~\footnote{https://bert.org/2020/10/01/pressing-yubikeys/.}
Other Robo projects, such as MattRobot\footnote{\url{https://www.mattrobot.ai}}, aim to simulate touches by employing Robo fingers capable of handling various types of touch interactions.
        
Once the attacker has registered their \hsk, they can mark {\em Remember Me} option to get the cookie which allows the victim to continue login while remaining oblivious that the attacker's \hsk was registered. If the victim logs in through a different browser or device they will get an error which is an indication that their \hsk was never registered. We answer the effectiveness of two email notifications in RQ1 in Section ~\ref{desc:RQs}.

\begin{table}[tbp]
    \caption{FIDO2 flaws exploitable by local attacks}
    \begin{center}
    \begin{tabular}{|p{4.5cm}|c|p{1cm}|}
        \hline
        \textbf{Flaws} & \textbf{Type} & \textbf{Attacks}\\
        \hline
        Lack of confidentiality/integrity of FIDO2 messages accessible to browser extensions & Protocol & 1\cite{hu2016security}, 2, 5 \\
        \hline
        Broken clone detection algorithm & Protocol & 7 \\
        \hline
        User misunderstanding from social engineering and notification/error messages & HCI & 1, 2, 3, 4~\cite{feng2021formal}$^\ddagger$ , 7 \\
        \hline
        Cookie life cycle & Implementation & 6\cite{patat2020please} \\
        \hline
    \end{tabular}
    \end{center}
    \begin{tabular}{l}
    \end{tabular}
    $\ddagger$~ = for FIDO UAF\\
    \label{flaw-table}
\end{table}

\subsubsection{Double-binding attack during registration \& authenticated session -- Attack 2}

\paragraph{During registration}
During registration of the victim's \hsk, Adversary A1 registers their malicious authenticator first to the victim's account and sends a second registration request in the background to register the victim's \hsk to the same account. 
The victim and the attacker can respond to their respective registrations if the \rp requires a touch for user presence on their \hsk. A1 directs the attacker's user presence request to the attacker. Like the mis-binding attack, the attacker can automate the test for user presence. 
Unlike the mis-binding attack, the victim can log in using their \hsk without even being aware that a second \hsk is also bound to the account.

A victim may detect one of the binding attacks (Attacks 1a, 2a, and 3a described earlier) in two ways unless A1 also compromises the notification channel.

\begin{itemize}[topsep=1pt,noitemsep,leftmargin=1.5em]
\item Some {\rp}s send a registration notification through a different channel to users, such as email. If users monitor these notifications, they can detect the attack when they see two notifications for attack 2a or an unexpected notification for attack 3a.
\item An \rp can display a list of registered {\hsk}s on the account, including the make and model. A user can verify the list regularly to detect a malicious \hsk. This passive detection depends on user awareness and effort.
\end{itemize}

We analyzed the notification process of ten popular web services that use FIDO2's WebAuthn for 2FA (see Table~\ref{rp-analysis}). Six services do not send any registration notifications. Four services send an email notification. However, all of them send an identical email when adding a new \hsk, as shown in Figure~\ref{fig:email_notification_and_clone_error_msgs_fig} (and in appendix, Figure~\ref{fig:user_confusion_errors_confusion}). 
For Attack 2a, receipt of two duplicate emails that lack any information about the identity and number of \hsk devices may not be sufficient to raise suspicion. We answer the effectiveness of two email notifications in RQ2a in Section ~\ref{desc:RQs}.

% \tarun{Mention RQs which measure this}

\paragraph{During authenticated session}
Adversary A1 registers a malicious \hsk to the victim's account during a logged-in session by sending a registration request and the response from their \hsk in the background without any victim's involvement.
This attack can be executed anytime during a session, providing ample opportunity to initiate an attack instead of being limited to only during registration.
Adding a malicious \hsk allows an attacker to access the account stealthily indefinitely unless the victim manually detects it through the registered \hsk page.
    
A strong defense against this attack is to require authentication using an already-registered \hsk before adding a new \hsk. Requiring the previous 2FA prevents software attacks that register their malicious \hsk in the background without user knowledge. Of the ten web services we analyzed, as shown in Table~\ref{rp-analysis}, four do not require any authentication while adding a new \hsk during a logged-in session.
Five services require only passwords to add a new \hsk even if there is already a registered \hsk, which allows software-based attacks to add a malicious \hsk. 
Twitter allows only one registered \hsk, which prevents a double-binding attack. However, limiting registration to one \hsk is not recommended as it prevents a user from registering an \hsk as a secure form of backup.
%\ks{Does this defense suggestion also apply to attack 2? Should we merge this with recommendations in 4.1? Future work - can the attacker social engineer the user to provide the touch for user presence? Prompt a re-authentication? Future work.}

\subsubsection{Synchronized login -- Attack 3}
\label{para:sync_login}
As described in Figure~\ref{FIDO2_synchronized}, while a victim logs into a website with a registered \hsk, Attacker A1 generates a login to another victim account that expects an \hsk to authenticate. The second login is invisible to the user except that A1 coerces the user to perform the user presence test, believing it to be for the first site.

\begin{figure}[htbp]
\centerline{\includegraphics[width=0.5\textwidth]{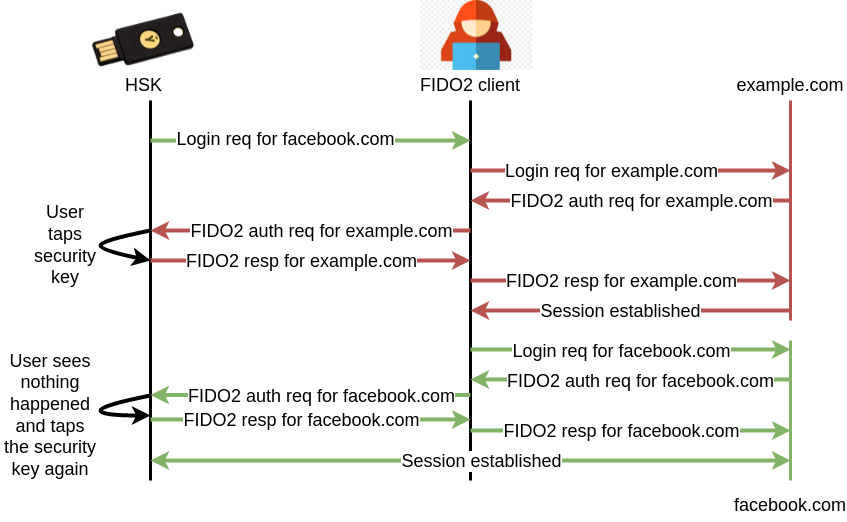}}
\caption{Synchronized login (Attack 6)}
\label{FIDO2_synchronized}
\end{figure}

To accomplish this, A1 adds an invisible iframe to the {\tt example.com} page that loads {\tt facebook.com}, resulting in an authentication request to the \hsk sent from the iframe. By default, the browser's webAuthn client does not allow cross-origin iframes to send an authentication request to an \hsk. To allow cross-origin iframe authentication, A1 also adds (1) an allow attribute to the iframe (i.e., \textit{allow="publickey-credentials-get *"}) and (2) a header in the response from facebook.com that loads the iframe (i.e., \textit{Permissions-Policy: publickey-credentials-get=*}).

There are two ways the test for user presence is handled. 
First, suppose {\tt example.com} has a registered \hsk but the user is not being prompted to authenticate using it. Many sites use a cookie for login once the user completes their first login from the device using their \hsk. The user could misinterpret the prompt for using the \hsk to login to {\tt facebook.com} in the background as a request to login to {\tt example.com} using the \hsk. 
Second, even if the login to {\tt example.com} requires authenticating with their \hsk, A1 can cause two prompts to tap the \hsk for user presence test and the user may comply believing the first tap failed to be recognized. 
Currently the browsers shows the domain of the website that a user is authenticating to, but users do not usually verify it. Furthermore, tapping \hsk twice is common due to improper touch first time. We answer these user behaviours in RQ4 in Section ~\ref{desc:RQs}.
% Currently, an \hsk does not receive an acknowledgment from an RP after successful authentication, so its hard for users to know that the authenticating multiple times could an an attack

This attack is more insidious than some other attacks because the attacker succeeds against the victim using the victim's \hsk without registering the attacker's \hsk, reducing the amount of forensic evidence for detecting the attack. 

This attack assumes that A1 has compromised the user's password to {\tt facebook.com}. 
If passwordless authentication is enabled on the target site, then A1 does not need a stolen password. 
In this case, passwordless authentication increases the risk of this attack on this account since A1 does not need to first compromise the password. 
A1 can compromise an account that the user has not accessed while A1 was present.

\subsubsection{MITM -- Attack 4}

When a victim visits a website, such as \texttt{facebook.com}, an attacker can manipulate the authentication process to execute a MITM. If the victim is already logged in, the attacker can forcibly log them out by removing the session details from the network request during communication with the website.

To carry out the attack, the attacker generates a session from their own device and initiates an authentication request to the victim's \texttt{facebook.com} account. The attacker intercepts the authentication request received by the victim and replaces it with the request they received on their device, aiming to obtain the victim's signature using their hardware security key (\hsk). The attacker then captures the response, transfers it to their device, and forwards it to \texttt{facebook.com}, enabling their device to create a login session on the victim's account. In case the user already had a session with \texttt{facebook.com}, the attacker adds the previous victim's session details to {\tt facebook.com} to allow the victim to log in using their previous session. 

In scenarios where the victim was not initially logged in to \texttt{facebook.com}, the attacker can either temporarily share their own session with the victim to avoid arousing suspicion or execute a variant of Attack 3. In this variant, the user is prompted to authenticate to \texttt{facebook.com} by tapping the hardware security key twice. One session is created for the victim, while the other is established for the attacker.

% A phishing website impersonates a website login page where a victim has registered an \hsk on their account. 
% Suppose A1's goal is to log in to {\tt facebook.com} as the victim. 
% A1 adopts the same exploit strategy as Attack 6 but relies on the phishing site to complete parts of the attack.
% The phishing page includes an invisible iframe for {\tt facebook.com} and adds the allow attribute to the iframe. 
% A1 only adds the header in the response from {\tt facebook.com} that loads the iframe  (i.e., \textit{Permissions-Policy: publickey-credentials-get=*}).
% The iframe generates an authentication request for {\tt facebook.com} to the \hsk, and A1 intercepts the response.

% The phishing website can drop the connection with the victim and impersonate them at {\tt facebook.com} or the phishing website can continue as an MITM attacker.
% Note that the phishing and target login pages can be for different sites where the user has a registered \hsk as long as A1 has also compromised the password if the target site requires it.

\subsubsection{Signature algorithm downgrade -- Attack 5}
A signature algorithm downgrade attack occurs during registration when adversary A1 modifies the list of signature algorithms in the registration request sent from the \rp to the WebAuthn client. In the registration request, the \rp sends a prioritized list of signature algorithms that it supports. The attacker leaves only the least secure algorithm and removes all others.
    
All ten web services in our analysis use secure signing algorithms that are not straightforward to crack. However, this attack could be used with other web services in the wild that support at least one insecure signing algorithm. Also from the ten RPs we analyzed, as shown in Table~\ref{rp-analysis}, seven support the inadvisable signing algorithm $RSASSA$-$PKCS1$-$v1\_5$ using $SHA$-$256$. An attacker can force the \hsk to use an inadvisable algorithm and possibly exploit it.
The $RSASSA$-$PKCS1$-$v1\_5$ using $SHA$-$256$ has been exploited in the past using Bleichenbacher’s attack, which allows one to perform arbitrary RSA private key operations. Given access to an oracle, and insecure exponents, the attacker can sign arbitrary messages with the \hsk private key. 
%In our work, we haven't explored the possible oracles available or how often it can be exploited. 
Table~\ref{rp-analysis} lists the minimum supported signing algorithms by individual RPs from our analysis.

The signature downgrade attack is more of a theoretical threat today. No RPs that we analyzed support weak signing algorithms. It is important to understand the threat exists in FIDO2 because experience shows that downgrade attacks are exploited after many years when algorithms become outdated.

\subsubsection{Cookie Lifecycle -- Attack 6}
Although 2FA strengthens security, it decreases usability when users must always authenticate using an \hsk. In response, some {\rp}s use cookies to avoid 2FA for future logins on a device if users select "Remember this device."

For example, Facebook issues a $`datr`$ cookie to denote the user has previously authenticated using an \hsk on their device. 
Future logins skip requiring the \hsk until the cookie expires.
If an attacker steals this cookie and adds it to the browser on their device, the attacker can authenticate as the user with only a stolen password and altogether bypass 2FA. 

The "cookies" permission allows a Chrome extension to use the cookies API. With API access, A1 can steal long-term session cookies or cookies that help them bypass 2FA on any device for a domain. We analyzed 246,345 Chrome extensions and discovered 31,326 Chrome extensions with "cookies" permission. 18,024 of those extensions also have "://*" permission, allowing these extensions to steal the cookies of every \rp a user uses. 102 of these extensions have more than a million users each. Figure~\ref{chrome_extensions_analysis_chart} contains the full distribution of users for these extensions.

%For example, Facebook uses a $`datr`$ cookie to identify the web browser used to connect to Facebook independent of the logged-in user. This cookie plays a key role in Facebook's security and site integrity features. If an attacker steals this cookie, it can bypass 2FA for that account on any device. These cookies stay in the user's browser until they expire, or the user explicitly removes the device from their account. This cookie drastically decreases the security of that account, as now software attacks could steal those cookies and use those cookies on other devices to bypass 2FA hardware authenticator authentication. It breaks the whole point of having an \hsk for 2FA.

% An attacker can also choose to steal session of a logged in session to a 2FA enabled Facebook account by copying cookies $`c_user`$ and $`xs`$ to a new device, without using 2FA. 
    
To test the feasibility of a cookie stealing attack, we logged into Facebook from a Chrome browser and enabled 2FA using an \hsk. We added the current browser to the "Remember this device" list. Then, we logged out of Facebook and copied the $`datr`$ cookie to a Chrome browser on another device. We confirmed that the cookie permitted us to successfully authenticate to Facebook from the second device with just a password and bypass 2FA.

Cookie-stealing attacks are well-known. From our experiment, we learned that Facebook sets cookie expiration at approximately two years. 
Although this is motivated by usability, attackers can exploit to to grant them long-term access to 2FA-enabled accounts without stealing the hardware authenticator.

\subsubsection{Bypassing clone detection algorithm -- Attack 7}
\paragraph{Lack of informative error message}
We experimented with triggering the clone detection error messages on ten RPs. A few sites do not report any errors. Most sites display a generic error message that is not specific to device cloning (see Table~\ref{rp-analysis}). 
These generic messages may lead users to reattempt the login. Multiple reattempts could eventually lead to a successful login, depending on the number of times the attacker used the cloned device. For example, if the \hsk and \rp have a counter value of $x$, the cloned \hsk would initially have the same value. When the attacker logs in, the counter's value updates to $x+1$ at the \rp and the cloned \hsk. If users try to log in to their account, they may receive an error message, but the counter on the user's \hsk will advance to $x+1$. If the user reattempts the login, they will be granted access to their account and may not understand the failed attempt is due to a cloned \hsk. We answer the effectiveness of two email notifications in RQ3 in Section ~\ref{desc:RQs}.

\begin{figure}[bht]
        \begin{subfigure}[b]{0.24\textwidth}
                \includegraphics[trim={11cm 2cm 11cm 4cm},clip,width=\linewidth]{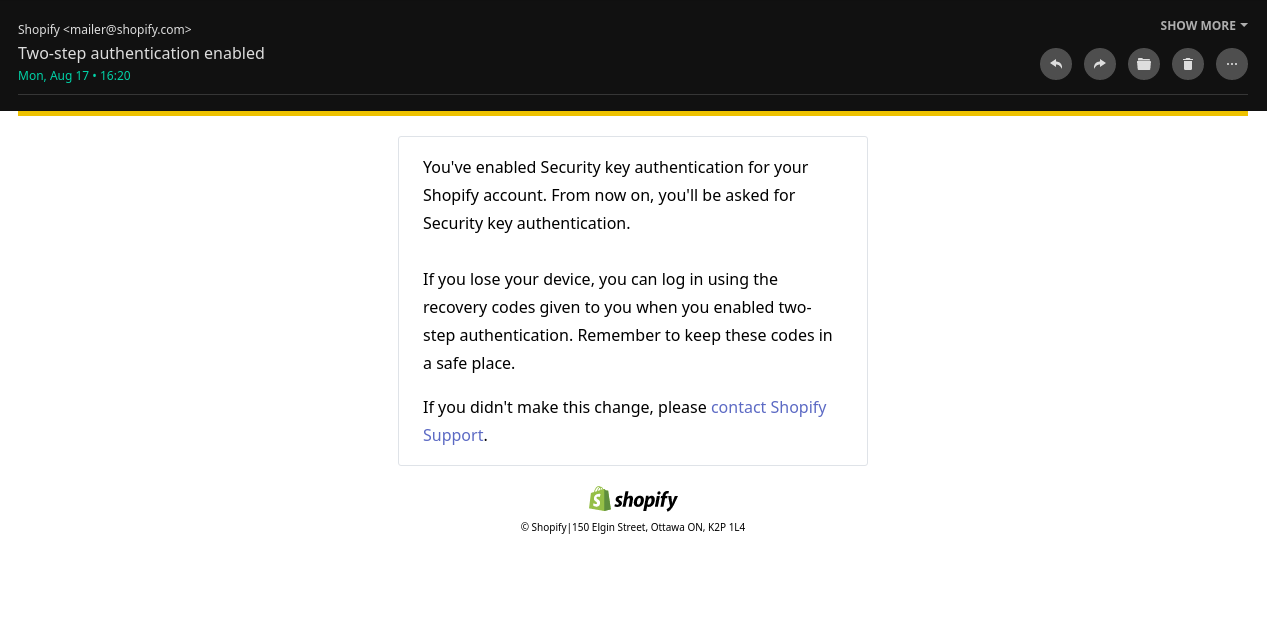}
                \label{fig:email_not_shopify}
        \end{subfigure}%
        \begin{subfigure}[b]{0.24\textwidth}
                \includegraphics[width=\linewidth]{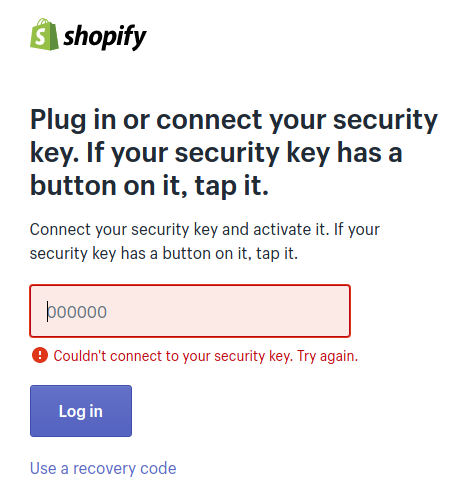}
                \label{fig:clone_error_shopify}
        \end{subfigure}%
        \caption{Email notification when adding a new HSK and clone detection error on Shopify}\label{fig:email_notification_and_clone_error_msgs_fig}
\end{figure}

Figure~\ref{fig:email_notification_and_clone_error_msgs_fig} (and Figure~\ref{fig:user_confusion_errors_confusion} in appendix) shows a screenshot of error messages from {\rp}s when the counter value submitted by the \hsk is lower than the value stored at the \rp. Even though this indicates a possible cloning attack, none of the error messages state that the user's account may be under attack. Several messages even encourage the user to switch to a less secure form of authentication.

\paragraph{Stealthy device cloning attack}
FIDO2 supports a counter to detect device cloning, as explained in Section~\ref{background}. 
Assume A2 has cloned a victim's device to gain unauthorized, undetected, short-term access to their account. The following paragraph describes a stealthy attack by A2 that avoids detection by the clone detection algorithm in FIDO2.

Assume the user's \hsk and the \rp have an account-specific counter value of $x$. When A2 gains physical access to the user's \hsk, they first clone the \hsk and then increment the counter value on the user's \hsk (by $y$) by sending it $y$ dummy authentication requests for the RP they want to compromise. The dummy requests can be created using libraries such as Yubico's libfido2. The user's \hsk counter is now $x+y$, the cloned \hsk counter is $x$, and the {\rp}'s counter is $x$. The attacker can now log in to the user's account up to $y$ times before the user logs in without detection. When the user logs in, they won't trigger an error if $x+y>x+(number Of Times Attacker Logged in)$, but the login will force an update to the RP counter to $x+y$. The attacker will not log in again without detection, but the attack provides a window of opportunity that the attacker can exploit. 
An attacker can perform a variation of this attack if the \hsk maintains a global counter.

\section{Attacks feasibility}
% In this section, we first describe our implementation and then analyze real-world browser extensions to identify the browser extensions that have the ability to execute these attacks.

To demonstrate the feasibility of the attacks, we built a prototype of a malicious Chrome extension that compromises a webAuthn client and executes the eight attacks. In our Chrome extension, content scripts allowed us to obtain details and make changes to the webpages a browser visits. We replace Chrome's web API function \textit{navigator.credentials.create} with our custom handler on every webpage. Our custom handler modifies/replaces the original FIDO2 request or response with a malicious one.
        
\begin{figure}[htbp]
\centerline{\includegraphics[trim={0 0 4cm 0},clip, width=0.5\textwidth]{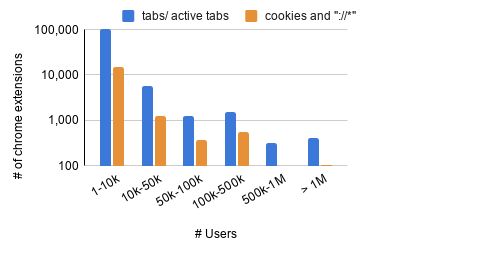}}
\caption{User distribution for Chrome extensions with permissions that allow a MITM attack against FIDO2 {\hsk}s}
\label{chrome_extensions_analysis_chart}
\end{figure}
        
Similar to Kaprevelos et at.~\cite{kapravelos2014hulk}, we analyzed the permissions requested for 246,345 Chrome extensions, extracted by CRXcavator~\cite{CRXcavator} from the Chrome webstore on Jan 21st, 2021. We found that 115,881 Chrome extensions use \textit{activeTab/tabs} permission, which is sufficient to execute a MITM attack by overriding webAuthn client APIs. There are 404 extensions with more than one million users each that can execute these attacks. A malicious actor only has to compromise one of these extensions to be in position to launch an attack on over one million users. Figure~\ref{chrome_extensions_analysis_chart} shows the distribution of users among these extensions.

Prior research demonstrates the feasibility of Attacker A2 obtaining a clone of an \hsk once they have physical access. For example, Roche et al.~\cite{roche2021side} describe how to clone a Google Titan Security Key. It requires about 10 hours to complete the cloning process. Assuming that cloning has taken place successfully, we describe how an attacker can avoid the clone detection algorithm described in the FIDO2 specification. 

The remainder of this section provides detailed descriptions of attacks that Adversary A1 can execute on the webAuthn client API and Adversary A2 can execute with a cloned \hsk.
Others have explored a few attacks against FIDO2 (see Section~\ref{related}). We take a holistic approach to explore a range of attacks in detail. We are the first to demonstrate the feasibility of these attacks.

\paragraph{Static and dynamic analysis of real chrome extensions}
We tried to find real-world attacks by analyzing extensions. We did not detect any attack happening in the wild as of now.

We analyzed extensions from the Chrome store to determine if any extensions in the wild were executing attacks on FIDO2. Specifically, we downloaded the source codes of 152,526 Chrome extensions in September of 2022 from the Chrome store over a period of 20 days to minimize additional load on the Google server. Our testing pipeline consisted of two stages: a static and dynamic analysis.

In the static analysis phase, we filtered extensions by two criteria: whether an extension had the ‘modify all data’ permission or if the phrases \texttt{navigator.get} or \texttt{navigator.create} were present in any of the source files. If either criterion was met, the extension was flagged and placed in a folder to be run through dynamic analysis.

In the dynamic analysis phase, we installed each flagged extension on a browser in a virtual machine in sequence. The virtual machine had a virtual FIDO2 authenticator that automatically approves registration and authentication requests. We ran an authentication server on the client outside the virtual machine. We performed complete FIDO2 registration and authentication flow between the webAuthn server and the virtual FIDO2 through the Chrome browser, where the browser extension under analysis was installed. For each request, we verified the Hash(ClientData) the server sent and the virtual authenticator received for registration and authentication. In response, we verified the public key and the signed object sent by the virtual key with what the server received.

To account for extensions that might trigger malicious code based on URLs, we ran complete registration and authentication flows for each extension on the top 100 Alexa domains. We added mapping of these URLs to our local server's IP address in the host file so that for all the URLs, the browser is still connected to our local server. We tested our setup on our proof-of-concept malicious extensions and were able to detect all the attacks.

Our static analysis identified 155 extensions that matched our criteria. However, none of the extensions in the dynamic analysis were detected as performing attacks on FIDO2. The extensions flagged by static analysis, particularly those using the \texttt{credentials.create} and \texttt{credentials.get} APIs, fall into two categories: password managers providing FIDO2 support and poorly optimized extensions. The first category of false positives pertains to  valid use cases such as password managers trying to implement 2FA for user's account, which were flagged for using the ‘credentials.create’ and ‘credentials.get’ APIs. The second category of false positives pertains to poorly optimized extensions. These extensions are created using webpack, which packs multiple dependencies into a single file. Sometimes parts of these dependencies are unused, and it seems they are included for an API call that never happens. One such dependency is ‘simplewebauthn’ which is depended on by other libraries but only makes calls to the ‘credentials.create’ and ‘credentials.get’ APIs not usable by extensions.
\section{User study}

Our user studies aim to assess whether the existing error messages, email notifications, and UX behavior enable users to detect various attacks. The results will show the effectiveness of the attacks under the current system design. 

We excluded the signing algorithm downgrade, cookie attack, and MITM attack from our study because they do not cause any changes to the UX. One variant of the MITM attack modifies the UX by presenting two pop-ups with the same domain name. This alteration in UX resembles a synchronization attack and possesses a higher degree of stealthiness. We can establish an upper bound for detecting the MITM attack variant by measuring users' perception of UX changes during synchronization attacks. We designed our studies to address the following four research questions, corresponding to one of the four primary attack types.

% \paragraph{Research Questions}

% One variant of MITM attack changes UX by displaying the two popups with the same domain name. This UX change is similar to synchronization attack and is more stealthier, so measuring user's perception of changes in UX during synchronizing attack will help us get the upper bound of detection

% \begin{description}
%     \item[RQ1] Double-binding-- How do users interpret the addition of a rogue \hsk they encounter on the settings page?
%     \item[RQ2] Clone detection-- How do users interpret clone detection error messages they encounter during the login process?
% \end{description}

\label{desc:RQs}
\begin{description}
    \item[RQ1] Misbinding-- How do users interpret the error and their inability to log in after a misbinding attack?
    \item[RQ2] Double-binding-- 
    \begin{enumerate}[label=(\alph*)]
        \item How do users interpret when they receive two registration emails after a \hsk registration?
        \item How do users interpret the addition of a rogue \hsk they encounter on the settings page?
    \end{enumerate}
    \item[RQ3] Clone detection-- How do users interpret clone detection error messages they encounter during the login process?
    \item[RQ4] Synchronized login-- 
    \begin{enumerate}[label=(\alph*)]
        \item How do users interpret pressing the HSK button twice before logging in? 
        \item Do they detect the attack by observing the browser's popup displaying the website name?
    \end{enumerate}

\end{description}

We conducted two user studies to evaluate the detectability of attacks given the current error messages and user experience (UX). Based on our experience, we hypothesized users would likely not detect these attacks in practice. Therefore, we decided to prime the participants during the study by asking them to watch for potential attacks that would take place for some of them. The results represent the best-case scenario for detecting the attack due to the priming.

The objective of the first study, an online survey, was to assess users' comprehension of error messages, notifications, and unusual UX flow. For this investigation, we selected RQ2a, RQ3, and RQ4a, which were amenable to measurement through an online survey and necessitated minimal contextual information. RQ2a and RQ4a were also measured in the subsequent in-lab study to ascertain their validity.

% The first study was an online user study aimed at assessing the effectiveness of error messages during a clone detection algorithm and email notification in a double-binding attack. These attacks were chosen for the online study due to their feasibility in an online context and the ability to involve a larger number of participants.

The second user study was an in-lab investigation that addressed all the research questions (RQs) except RQ2a. We developed a malicious browser extension that executed the attacks to provide participants with the necessary contextual information. We then had participants register an \hsk and log in to a test account. At the same time, the extension performed the attacks, allowing us to capture participants' perceptions of the events that transpired during an attack.

% error messages in the context of misbinding, clone detection, the presence of two security keys on the settings page after a double binding attack, and the domain names and the requirement of a double tap for a synchronized login attack.

\subsection{Survey study}

% section presents the methodology and findings of the online study conducted. The objective of this study was to assess users' comprehension of error messages and notifications by engaging a larger number of participants compared to an in-lab study. For the purposes of this investigation, we selected only two research questions (RQs) that were amenable to measurement through an online survey and necessitated minimal contextual information. Nevertheless, RQ2 was also measured in the subsequent in-lab study to ascertain its validity and consistency.

% This section provides the methodology and results of the online study. In this study, we wanted to measure the user's understanding of error messages and notifications over relatively large number of participants than we could in-lab study. We only chose these two RQs for this as these could be measured in an online survey and requires minimum context. However, we again measure RQ2 in-lab study to get see if provides. 
% \ks{Can we say more about how the survey was developed? Our hypothesis going in? }

We conducted an IRB-approved survey exploring participants' understanding and actions when encountering attacks 2a, 3, and 7a. We asked participants open-ended questions about what they understand and their next step in the following scenarios: (1) [RQ2a] Attack 2a- receive two \hsk registration emails due to double authenticator registration, (2) [RQ3] Attack 7a- encounter clone detection error message, and (3) [RQ4a] Attack 3- requires two taps to complete authentication. 
% The complete survey is available in Appendix~\ref{appendix:user_study}.
% \ks{Earlier we mention the RQs we will explore. Then this section talks about the attacks. Could be confusing. Which attack goes with which RQ? Should we mention that?}

\paragraph{Demographics} We recruited n=80 participants from Prolific. We mentioned in our study page \textit{``Please only take part in this survey if you have used a security key or hardware token such as yubikey for authentication.''} However, according prolific policies we cannot screen for surveys using its description and therefore we could not reject any participants. Out of 80 participants, 32 use \hsk for their accounts, 45 do not use \hsk for their accounts, and 3 were unsure. Table ~\ref{table:Survey_demographic} presents the survey demographics. The median time to complete the survey by participants was 6 mins, and we paid them each \$1.20 \ie \$12/hr. 

\begin{table}[t!]
\centering
\small
\begin{tabular}[t]{lr}

\textbf{Metric} & \textbf{Percent} \\
\toprule
\textbf{Gender} \\
%   \rowcolor{gray!50}{1}
%   \rowcolors{2}{gray!25}{white}
Male & 50 \\
Female & 50 \\
\midrule
\textbf{Age} \\
18-29 years & 38.8 \\
30-39 years & 32.5 \\
40-49 years & 10\\
50-59 years & 12.5` \\
60+ years & 6.2 \\
\midrule
\textbf{Student status} \\
Student & 31.3 \\
Non Student & 66.3 \\
% Unknown & 24.7 \\
\midrule
\end{tabular}
\begin{tabular}[t]{lr}
\textbf{Metric} & \textbf{Percent} \\
\toprule
\textbf{Ethnicity} \\
White & 80 \\
Asian & 10 \\
Black & 6.2 \\
Mixed & 3.8 \\
% Unknown & 14.4\\
\midrule
\textbf{Employment} \\
\textbf{Status} \\
Full-time & 56.3\\
Part-time & 20\\
Unemployed & 12.5\\
Unpaid work & 5\\
% Unknown & 34.4 \\
\bottomrule
\end{tabular}
\caption{Survey participant demographics. Percentages may not add to 100\% because we do not include “Other” or “Prefer not to answer” percentages for brevity.
}
\label{table:Survey_demographic}
\vspace{-3mm}
\end{table}

\subsubsection{Methodology}
We designed three questions corresponding to attacks 2a, 3, and 7a.

\paragraph{[RQ2a] Attack 2a}
To measure users' reaction to Attack 2a, we described a scenario where they recently registered their \hsk with Github. We provided them with credentials for a test Gmail account containing two \hsk registration emails to simulate attack 2a and other random emails such as GitHub account creation and a 2FA enrollment email. Figure~\ref{fig:github_notification} shows the content of the email they see. We tasked each participant with logging into the test Gmail account and identifying whether their Github account had any malicious activity. Then, we primed participants by telling them that half of the participants' GitHub accounts had some malicious activity. We introduced this priming to determine if even cautious users will ignore two consecutive \hsk registration email notifications arriving within seconds of each other.

\paragraph{[RQ4a] Attack 3}
To explore users' behavior during Attack 3, we described a scenario where they had used an \hsk with their work account for a long time. We told them that one day while trying to log in, they tapped their \hsk to satisfy the user presence test. It doesn't work the first time, but it works after tapping it the second time. We then asked them why this might occur and how they would respond. 

\paragraph{[RQ3] Attack 7a}
To explore user behavior during Attack 7a,  we described a scenario where they have been using an \hsk with their work account for a long time, and suddenly one day, they get an error while logging in. We showed participants GitHub's clone detection error message as shown in~\ref{fig:user_confusion_errors_confusion} and asked what they thought was the reason for the error message and what would be their next step.

\begin{figure}[htbp]
    \centering
    \begin{subfigure}[b]{0.47\textwidth}
        \centering
        \includegraphics[trim={0 3cm 0 3cm},clip, width=\textwidth]{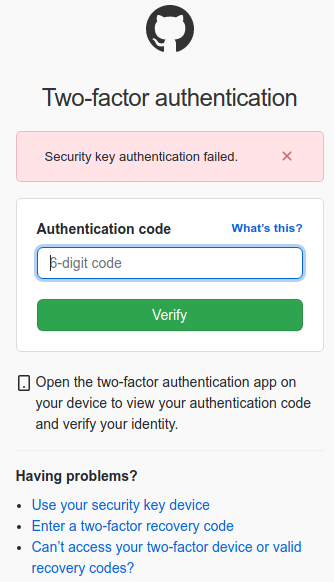}
        \caption{Github error notification for every failed login including clone detection error}
        \label{fig:github_error}
    \end{subfigure}
    \hfill
    % \vspace{20pt}
    \begin{subfigure}[b]{0.47\textwidth}
        \centering
        \includegraphics[trim={0 0 2cm 16.5cm},clip, width=\textwidth]{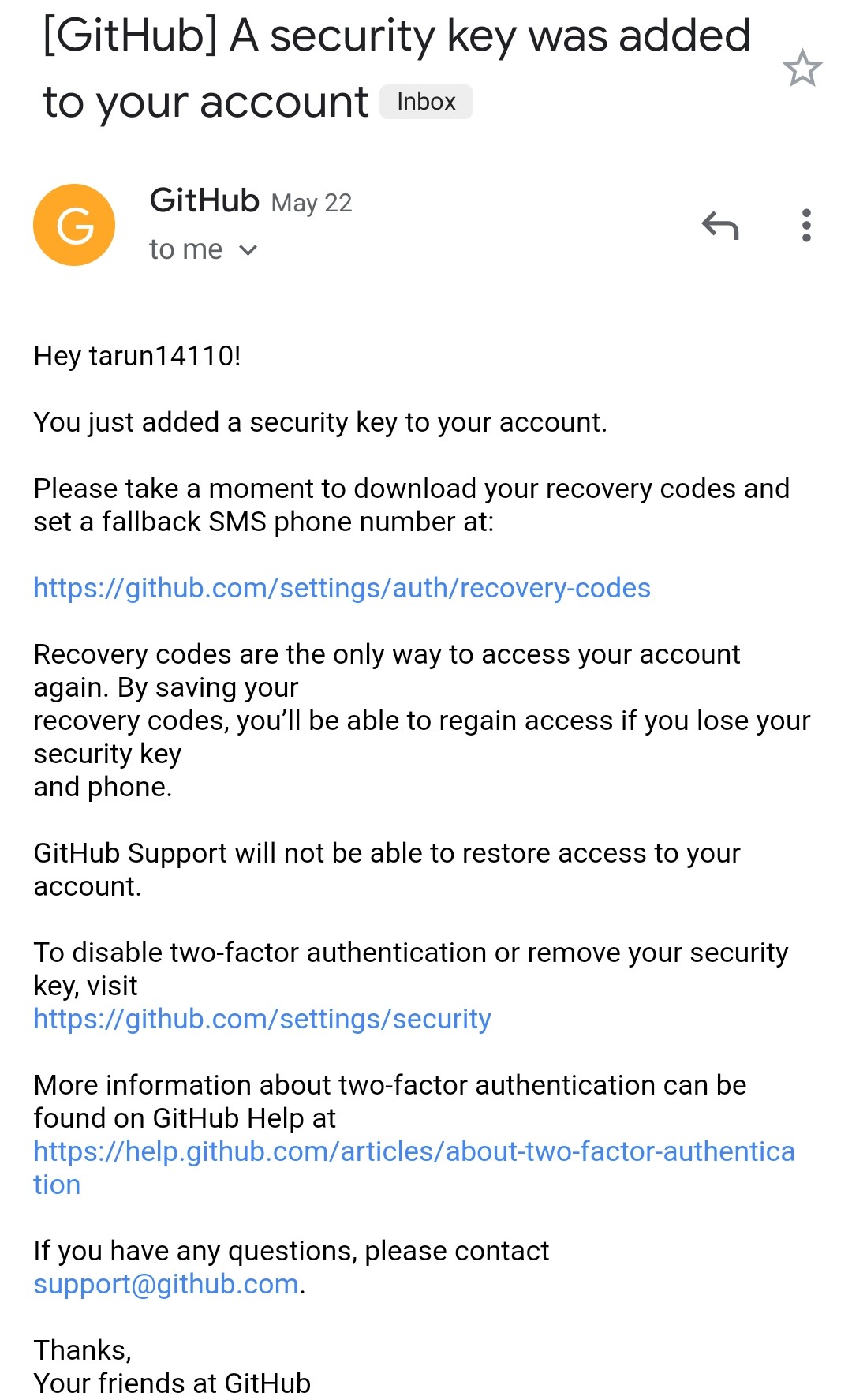}
        \caption{Email sent to users on registration of a \hsk}
        \label{fig:github_notification}
    \end{subfigure}
    \caption{Github error message and emails were shown to users for online study.}
    \label{fig:github_error_notificatsions}
\end{figure}

\subsubsection{Results}
%\tarun{Expand results}
\paragraph{Attack 2a}
None of the participants identified malicious behavior with their GitHub account, even with the priming. Five participants noticed the two \hsk registration emails but did not consider it as malicious activity. One of the reasons we believe most participants did not notice two emails is because Gmail by default hides the content of the second email if the content is the same as of the first email. Furthermore, even if somebody notice two emails two identical emails within seconds can be considered as an intential notification mechanism or a configuration error, but it is hard to imagine as an attack.

\begin{quote}
    P11: \textit{``It doesn't seem like there was any suspicious activity. There were only 3 emails and 2 of them talked about 2-step authentication processes.''}
\end{quote}

\paragraph{7a}
% No participant considered it an attack. The most common potential reasons participants mentioned were "The \hsk is not plugged in correctly" and "Dust on hardware token or USB slot." Table~\ref{table:user_study_themes} in the Appendix summarizes all the reasons and the next steps mentioned by participants in the survey. Below are some example quotes:
We ask participants what do they think the reasoning of the error. No participant considered it an attack.
The reasons for errors participants presented include incorrect key connection, wrong hardware token, USB reading issues, dust on hardware or USB slot, corrupted security key or PIN, device synchronization problems, incorrect username or password, prompt interaction delay, security token expiration, computer not updated, additional layer of authentication, and server errors. Based on these reasons users expressed they would perform following steps: unplugging and reconnecting the key, refreshing and retrying with the correct hardware token, using alternative 2FA methods, cleaning the USB port, resetting the security key PIN, re-entering username and password, restarting the computer, seeking help on Google for username/password issues, and  contacting GitHub support for prompt interaction delays.

\begin{quote}
    P23: \textit{``Could be a number of things. Could be the two-factor key got zapped or erased somehow. The USB port is not working. The authenticator key is dirty, and the contacts need to be cleaned with isopropyl. Something is corrupt in the computer. Needs rebooted.''}
\end{quote}
\begin{quote}
    P65: \textit{``Sometimes websites have temporary bugs that are beyond the user's control.''}
\end{quote}

\paragraph{Attack 3}
No participant considered that it was malicious behavior. Instead, participants mentioned that the \hsk tap is unreliable and may not work the first time. Therefore, they would tap it again. If that did not work, they would follow the same steps mentioned for the clone detection error messages.

\subsection{In-lab study}
We conducted a second IRB-approved in-person study to understand whether users detect these attacks when provided with the context they typically receive during an actual attack.

\subsubsection{Methodology}

To address the four research questions, we implemented the corresponding attacks in a browser extension. Participants were assigned the following four tasks, each corresponding to one of the research questions:

\textit{Task 1:} Participants were presented with a scenario where they needed to register an \hsk for their work \texttt{github.com} account and verify it by logging in. Our attack implementation ensured that the user experience remained unchanged during registration. Participants observed an error message during login. 

\textit{Task 2:} Participants were presented with a scenario where they needed to set up an \hsk with their \texttt{github.com} account. During the \hsk registration process, we also registered a fake \hsk with the nickname \textit{admin}. Participants were asked to navigate to the settings page to verify if the \hsk was registered correctly and suggest improvements for its security.

\textit{Task 3:} Participants were presented with a scenario where they needed to log in to their work \texttt{github.com} account, which they had been using for years. Upon the first login attempt, they encountered a clone detection error message. However, a subsequent login attempt succeeded due to an increment in the authentication counter after the failed attempt.

\textit{Task 4:} Participants were presented with a scenario where they needed to log in to their work \texttt{github.com} account, which they had been using for years. During login, our extension sent an authentication request on behalf of \texttt{chase.com} and another authentication request for \texttt{github.com}, as shown in Figure~\ref{fig:github_task4}. Users were required to tap the \hsk button twice, and before each popup, they observed the browser's popup displaying the domain they were authenticating to.

\begin{figure}[htbp]
    \centering
    \begin{subfigure}[b]{0.45\textwidth}
        \centering
        \includegraphics[width=\textwidth]{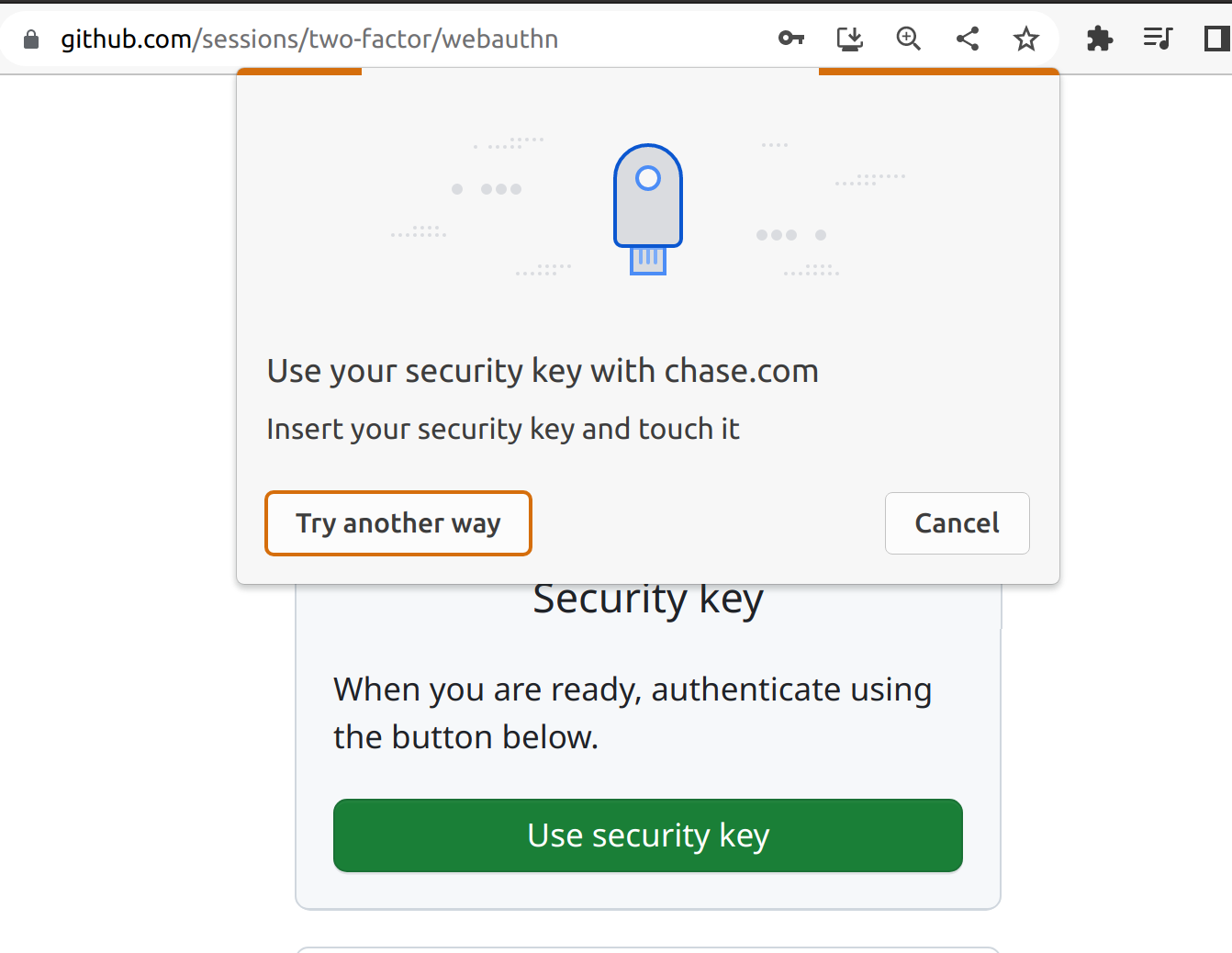}
        \caption{Chase.com authentication popup}
        \label{fig:figure1}
    \end{subfigure}
    \hfill
    \vspace{20pt}
    \begin{subfigure}[b]{0.45\textwidth}
        \centering
        \includegraphics[width=\textwidth]{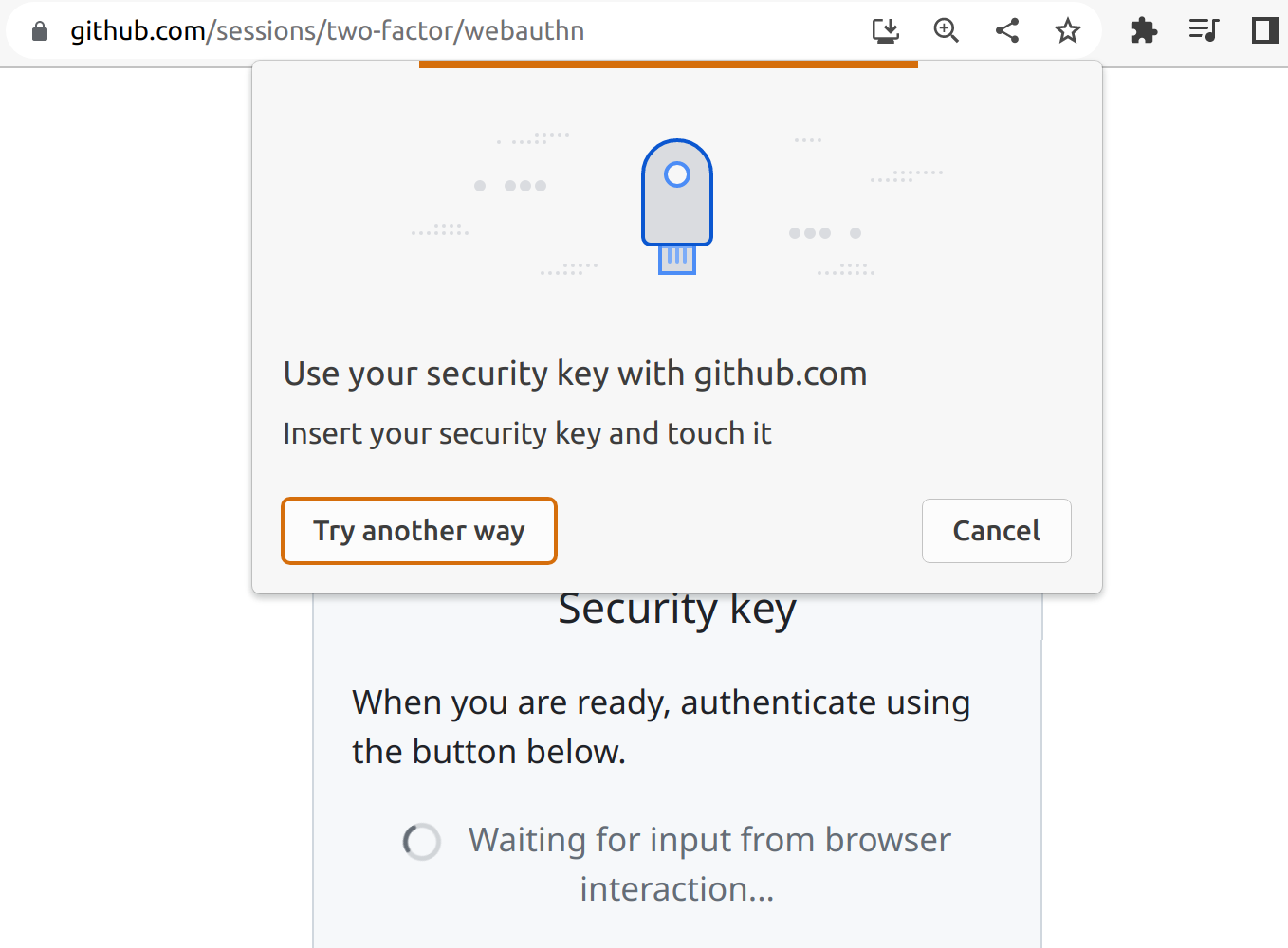}
        \caption{Github.com authentication popup}
        \label{fig:task4_githb_login}
    \end{subfigure}
    \caption{Task 4: Synchronization login of chase.com while user logs into github.com}
    \label{fig:github_task4}
\end{figure}

At the beginning of the study, participants were informed about the four tasks they would do and how Yubikeys worked. Participants were directed to a webpage where they received detailed instructions for each task as they progressed through the study. The browser extension automatically detected the current task and executed the corresponding attack. We used four separate GitHub test accounts to maintain isolation between tasks. This automated setup allowed us to minimize the study coordinator's interaction with the participants and ensured privacy. We believed that providing participants with this autonomy would result in behavior similar to their real-world actions. We instructed participants to try and resolve any encountered errors to the best of their abilities and provide their reasoning for why the errors occurred and what steps they could take to address them.

\textbf{Recruitment}
We aimed to recruit technically proficient individuals to maximize the likelihood of detecting the attacks. We recruited 20 participants by making announcements in Computer Science department classes, utilizing Slack, and distributing flyers in the CS department building. To facilitate communication, we specifically selected participants who were fluent in English. Additionally, we restricted the study to individuals who had prior experience using GitHub. Initially, we offered a compensation of \$10 for an expected 30-minute duration. However, we faced difficulties in obtaining sufficient participants. Consequently, we increased the compensation to \$20 during the study and paid this amount to all 20 participants. On average, participants took 22 minutes to complete the study.

\textbf{Demographics}
Out of 20 participants, 13 participants fell into the age range of 18-24, while 7 participants were in the age range of 25-34. In terms of gender distribution, 14 participants identified as male, while 6 participants identified as female. Additionally, 15 of the participants were undergraduate students, while 5 participants were graduate students.

\textbf{Data collection and analysis}
During each task, participants were questioned about their comprehension of encountered errors and their approach to resolving them. In addition, their screen activity was recorded to facilitate later analysis of their process. The research coordinator took notes if participants made any relevant comments related to our research questions. The collected responses were analyzed using inductive coding and content analysis techniques.

\subsubsection{Results}

\paragraph{RQ1}
% \ks{Broken reference in this paragraph}
For task 1, all participants successfully registered Yubikey on a GitHub account. All the participants saw the error message as shown in Figure ~\ref{fig:github_error} during login, which confirms that our attack was succesful for each of them. As expected, none of the participants realized they were logging in with a different key than the one registered with their account. Participants provided various reasons for being unable to log in, including broken keys, improper keypresses, and lack of familiarity with the key. Three participants using the Yubikey for the first time mentioned specific reasons such as improper fingerprint scan and failure to unplug the Yubikey after registration. We believe such misunderstandings would not apply to Yubikey users in general, as they would gain more experience and understanding over time.

\paragraph{RQ2b}
In task 2, participants were asked to register Yubikey on a Github account. In the background, we registered a second \hsk with the nickname \textit{admin}. Even when explicitly asked to check the settings page after login and verify the registered Yubikey and security of the account, only three out of twenty participants noticed that there were two Yubikeys registered. Of these three participants, only one mentioned that they would try to identify and remove the extra Yubikey. The other two participants did not comment on whether they would remove it, but we assume they would take some action. Therefore, we observed three successful detections of the double-binding attack.

\paragraph{RQ3}
In task 3, participants were asked to log in to GitHub using their Yubikey. The attack setup provided a cloning detection error during login on the first attempt, but upon a second attempt it successfully logged in. None of the participants associated the error with a potential cloning detection attack, as the error message did not provide any indication of such an attack. The error message was identical to other error messages, such as using the wrong Yubikey. Nine out of twenty participants attempted a second login and were able to log in successfully. All participants were unsure why the login failed the first time and attributed their inability to log in to an error in the authentication process. The rest of the participants mentioned they would log in using either SMS or another recovery mechanism and re-register the key as some issue could have happened with the registered \hsk. Despite our mentioning that they have been using the same key successfully for a long time, some participants still attributed the issue to improper Yubikey configuration on their account. P4 even acknowledged the strangeness of the situation but could not connect it to any specific attack due to the lack of informative error messages displayed by GitHub.

\begin{quote}
    P2: \textit{``It worked the second time so it might have just been an error in when I inserted the key.''}
\end{quote}

\begin{quote}
    P4: \textit{``Maybe the yubikey isn't configured correctly. I would log in with sms and check that it is configured correctly. But if I've hypothetically used it a bunch before it doesn't seem like that would be the issue. I would probably reconfigure the security key.''}
\end{quote}

\paragraph{RQ4}
In task 4, participants were required to log in to their GitHub accounts. In the background, our extension authenticated first with a \texttt{chase.com} account and then with GitHub. Participants saw two popups from the browser for each of these authentications and had to tap the Yubikey twice to complete the login process. One participant encountered a technical error during the setup and was unable to log in, resulting in a sample size of 19 for this task. Among the 19 participants, all successfully logged in to their GitHub accounts by tapping the \hsk twice. When explicitly asked if they noticed anything unusual during the authentication process, 6 participants reported no unusual observations. Thirteen participants mentioned that they had to tap the Yubikey twice, with 9 of them attributing it to either an improper key tap or a glitch in the protocol or website. Only 1 out of 13 participants noticed the first popup prompt for "chase.com" but proceeded with authentication anyway. Two participants out of the 13 repeated the task upon seeing the question "Did you notice anything unusual?" and then noticed the \texttt{chase.com} prompt, mentioning that they would not have noticed it if it hadn't been explicitly asked.

\section{Clone Detection Algorithm}
    We propose a cloning detection algorithm for account-specific counters that is not vulnerable to the stealthy cloning attack (as described in Attack 5a).
    The \hsk (see Algorithm~\ref{alg:clone-hsk} in Appendix) saves a hash of the challenge it receives during registration as $hashC$.
    Upon receipt of an authentication request, the \hsk includes $hashC$ in the response and updates $hashC$ with a hash of the authentication request challenge.

    % \ks{Missing clone detection algorithm. Are we going to add this back to appendix?}

    The \rp (see Algorithm~\ref{alg:clone-RP} in Appendix) maintains a linked list of hashed challenges it sends to the \hsk in $hashList$. 
    The \rp initially creates the list and adds $hash(request.Challenge)$ to the list when registration completes (line 7). After sending the request, the \rp adds the hashed challenge to the hashList.
    Upon receipt of an authentication response,
    the \rp searches the list to determine whether it contains the hashed challenge in the response (line 10).
    If it finds no matching challenge, then a cloning attack is detected (line 11).
    If the hash is in the list, it removes items from the head of the list up to and including the matching challenge. 
    
    After cloning occurs, whichever \hsk logs in first (victim or attacker) will submit a matching challenge hash and proceed, while the other \hsk will fail its next login attempt.
    Whether the victim or attacker logs in last, the \rp  is alerted that an attack may have occurred and can take precautions.
    % to protect the user's account. 
    If the victim logs in last, they will receive a clone detection alert message informing them an attack may have occurred so they can take appropriate action.
    
    %\ks {Issues for discussion or future work: How should the RP react to a detection? (1) Notify the user out-of-band to determine if a compromise actually occurred. (2) Re-enroll the device to make sure the cloned device is locked out? Ask the user to login again - if denied then the account was compromised.}
    
    %\ks{Notifying the attacker that the attack was detected isn't helpful and gives information to an attacker. Should the error message be purposely vague and the RP communicates with the real user out of band?}
    
    The hash-based clone detection algorithm is backward compatible with  FIDO2. It avoids false positives when messages are lost or delayed.
    Suppose an authentication request \mbox{(\rp->\hsk)} is lost. Assume that \rp and \hsk each store hashes of $challenge_{i-1}$. When the \rp sends $challenge_i$, it adds the hash to the tail of the list. If that challenge never arrives at the \hsk, the next request will contain $challenge_{i+1}$, which is added to the list. The \hsk will return a hash of $challenge_{i-1}$ in the response and store a hash of $challenge_{i+1}$. The \rp will match $challenge_{i-1}$ in the list and remove it from the head of the list along with any earlier challenges still in the list. The hashes of $challenge_i$ and $challenge_{i+1}$ are still in the list.
    
    Suppose an authentication response \mbox{(\hsk->\rp)} is lost.
    Assume that \rp and \hsk each store hashes of $challenge_{i-1}$. When the \rp sends $challenge_i$, it adds the hash to the tail of the list. The \hsk will return a hash of $challenge_{i-1}$ in the response and store a hash of $challenge_i$. If the response never arrives, the \rp will send $challenge_{i+1}$ in the next request and add the hash to the tail of the list. The \hsk will respond with $challenge_i$ in the response, and store a hash of $challenge_{i+1}$. When the \rp receives a hash of $challenge_i$ in the response, it will remove the hashes of  $challenge_{i-1}$ and  $challenge_i$ from the head of the list. 
    
    The likelihood of a false negative is negligible because the attacker authenticating with the cloned device would have to guess a correct challenge hash with a probability of 1 in $2^{256}$. 
    The following claim defines the security property of the new algorithm to prevent and detect attacks.

\begin{theorem} \label{thm:clone_detection}
Let a user $u$ register ${\hsk}_u$ on a relying party \rp for authentication. 
Assume an attacker \A has physical access to ${\hsk}_u$ from time $t$ to $t'$, where $t <= t'$. 
Let $u$ authenticate to \rp at time $t_i$, where $t_i < t'$, and at time $t_j$, where $t_j > t'$.
If \A clones u's \hsk (\ie~${\hsk}_c$) by time $t'$, then one of the following occurs:

\begin{enumerate}
    \item Both $u$ and \rp detect the clone attack during  authentication at time $t_j$.
    \item \A cannot authenticate to RP as u.
\end{enumerate}

\end{theorem}

\begin{proof}[\unskip\nopunct]
\textit{Proof Sketch:} Clone detection relies on a Hashed Challenge List that the \rp maintains. The ordered list of hashes of challenges is a sliding window of challenges the \rp has sent to the \hsk for authentication. The \hsk returns the challenge saved from the most recent authentication, which an \rp will accept only once and then remove from the list.

Once the device is cloned after time $t'$, \A can authenticate either before or after $u$ authenticates at time $t_j$. Case 1 describes what happens when \A authenticates before $u$, and Case 2 describes what happens when \A authenticates after $u$. 

In Case 1, both ${\hsk}_u$ and ${\hsk}_c$ have a copy of the hashed challenge received at time $t_i$, so \A successfully authenticates to \rp as $u$ and the hashed challenge at time $t_i$ is removed from the list at the \rp. 
When $u$ authenticates at time $t_j$ and returns the hashed challenge from time $t_i$, the \rp detects the cloning attack and notifies the user.

The only way for \A to bypass detection is to somehow force $u$'s authentication at time $t_j$ to be successful. For this, \A would need to predict the last challenge \A receives before time $t_j$ and send that challenge to ${\hsk}_u$ during the cloning process between time $t$ and $t'$, which is infeasible because the challenges are generated randomly for every authentication. 

In Case 2, both ${\hsk}_u$ and ${\hsk}_c$ have a copy of the hashed challenge received at time $t_i$, so $u$ successfully authenticates to \rp at time $t_j$, and the hashed challenge at time $t_i$ is removed from the list at the \rp. 
If \A attempts to authenticate after time $t_j$ and returns the hashed challenge from time $t_i$, the \rp detects the cloning attack, and \A fails to impersonate $u$.
To impersonate $u$, \A needs to guess the challenge the \rp returns to $u$ at time $t_j$, which is infeasible since challenges are randomly generated. Therefore, in this case, \A cannot authenticate to the \rp.
\end{proof}

\section{Recommendations}

\paragraph {Include and highlight nicknames, make \&\ model in email notifications} \rp{s} should send a notification to the owner of the account after every \hsk registration, highlighting the nickname, make \&\ model of the \hsk and the total number of registered {\hsk}s. This notification may help users detect the registration of an unrecognized \hsk in Attack 2. The notification could also be effective against Attack 1 if the adversary uses an \hsk with a different make and model than the victim.

\paragraph{Require \hsk authentication before adding a second \hsk}
An authenticated user can usually register additional {\hsk}s. But an adversary can register a malicious \hsk in the background without user knowledge \textit{i.e.} Attack 3b. Therefore, an \rp should require authentication with a registered \hsk before registering an additional \hsk. However, this presents a problem when a user loses their \hsk and wants to register a new one. In this case, the user must remove the lost \hsk during an already logged-in session using less-secure authentication methods like a password. The requirement to either (1) authenticate with a previous \hsk or (2) remove a lost \hsk ensures that either the attacker cannot register their \hsk or the victim detects the removal of their \hsk on the account.

% \ks{This item is unclear to me. Are you saying the RP should require authentication with an HSK before adding a new HSK? But doing this likely means you have to have weak fallback mechanism to remove a device. So is that a showstopper for the first suggestion? Or do we only allow weaker fallback when removing devices? Does the latter mean you can only have one HSK registered at a time, and you must remove one before adding another? It seems this raises issues but doesn't come out with a clear recommendation. Perhaps that is intentional---we lay out the issues and designers can choose?}
% \tarun{you can add more than one \hsk, but the second \hsk can be only registered after authenticating with the first \hsk. If a user lost the first \hsk and cannot authenticate, they can remove the first \hska and register with a fallback mechanism. An attacker can always remove the first \hsk to register theirs but the UX changes and some error messages can be displayed to allow a user to make an informed decision. This is not a done deal and requires further research on how to convey to users that if they fail to log in what does it mean and what exact error messages \rp can show them?}
% \ks{I edited the paragraph. Check it and then delete all this commentary.}

\paragraph{Provide more specific context in error messages}
To mitigate Attack 7a, device cloning error messages should explicitly state that (1) the device is cloned, (2) remove the device from the account, and (3) register another \hsk.  
The failed authentication following device cloning could happen to either the victim or \A, depending on who authenticates last following a cloning attack. Future work could analyze an \rp notifying a victim out-of-band to make them aware the attack occurred when \A authenticates last.

When a user attempts to log in with a different key than the one registered, \rp{s} should display an error that explains the problem to the user so they better understand it and can take action accordingly. Also, an informative error message could potentially help users detect misbinding attacks. Finally, a clear error message could improve usability for users with multiple \hsk{s}.

\section{Related Work}
\label{related}

% \st{Previous research ~\cite{panos2017security}, ~\cite{hu2016security}, ~\cite{jacomme2018extensive}, \cite{pereira2017formal}, ~\cite{lyastanifido2} explored security guarantees of FIDO's U2F and UAF using formal analysis. Recent work ~\cite{chenprovable}, ~\cite{guirat2018formal} prove the security guarantees FIDO2 provides. In our work, we include unexplored and stronger adversaries in our threat model, which leads to several realistic attacks. We explore attacks on registration, cookie stealing, device cloning, and attacks that could happen due to human errors. To get a realistic idea of our attacks, we implement a prototype of our attacks and analyze the current implementations of FIDO2 for ten popular web services for the same. We propose design changes using the Trusted Platform Module TPM to defend against the strong adversaries in our threat model.}
Prior work has investigated the security guarantees of the first FIDO protocols---Unified Authentication Framework (UAF) and Unified 2nd Factor (U2F). UAF is a passwordless authentication system, and U2F is a standardized 2FA system.  Hu et al. performed the first informal analysis of UAF and outlined three attacks \cite{hu2016security}. They identify the mis-binding attack where an attacker binds their authenticator to a \rp instead of the user's authenticator. 
% \tarun{They also present parallel session attack on UAF under threat model of compromised webAuthn client. Our schyncronized login, }

Panos et al. also analyzed UAF \cite{panos2017security}. They identified attack vectors that lead to system compromise. The threat model includes attackers with access to the authenticator and control of the UAF client.

Peirera et al. performed the first formal analysis of FIDO1 \cite{pereira2017formal}. Their threat model includes a network attacker and an attacker that compromises the client or server. They verified the security of protocol as long as the \rp is verified. The analysis investigated only authentication, not registration.

Feng et al. \cite{feng2021formal} conducted a comprehensive analysis of the FIDO UAF protocol, confirming the mis-binding attack and identifying the presence of a parallel session attack. Building on their findings, we further demonstrate the practical feasibility of the parallel session attack in FIDO2, which we call MITM - Attack 4.

{\bf FIDO2/WebAuthn} 
Eventually, UAF and U2F merged into the W3C standardized protocol, FIDO2, which supports both passwordless authentication and 2FA. Guirat et al.~ \cite{guirat2018formal} present a formal proof of the protocol, with a threat model of a passive and active network attacker. They did not consider a compromised client. 
Finally, Jacomme et al. performed a formal analysis on several multi-factor authentication schemes, including FIDO2 \cite{jacomme2018extensive}. The threat model includes malware, fingerprint spoofing, and human error. FIDO2 was shown not to be provably secure against malware.  

Other studies \cite{ciolino2019two, das2018johnny, lyastani2020fido2} have investigated the usability challenges and perceived benefits of FIDO2 device usage. Alqubaisi et al. \cite{alqubaisi2020should} evaluate the threats of password attacks and compare passwords to single-factor FIDO2. Their findings suggest that single-factor FIDO2 performs better against the password-based threat model, but they do not address targeted attacks on the FIDO2 protocol.

% \st{Alam et al. cite{alam2019poster} identified issues developers may have when implementing webAuthn, including existing issues with documentation and incorrect mental models.} 

Chang et al. \cite{chang2017making} exposed the weaknesses of U2F to side-channel and MITM attacks and propose a modification to the U2F protocol to mitigate side-channel attacks. O'Flynn \cite{o2019min} also expands the threat model by showing an attack on {\hsk}s through Electromagnetic Fault Injection, which led to recovering secret data. Dauterman et al. consider the threat posed by hardware backdoors to {\hsk}s and propose True2F~\cite{dauterman2019true2f}, a strengthened version of U2F. True2F  protects against a malicious \hsk and increases the protection against token fingerprinting. While this work does consider a compromised browser, it asserts that if the True2F token behaves faithfully, it is no less secure than traditional U2F. Our work is complementary, and makes {\hsk}s more secure than traditional U2F assuming a compromised browser.

%\paragraph{Signature Algorithm Vulnerabilities}
%The FIDO2 protocol allows six different signing algorithms for authenticating messages. Among these is PKCS v1.5 which has been found to be vulnerable to attack \cite{izu2007analysis} , and \cite{bardou2012efficient}. Bock et al. \cite{bock2018return} was able to exploit the PKCS v1.5 vulnerabilities to sign and decrypt messages with a key from a vulnerable server. 

{\bf Remember this Device (RTD)} Patat et al. \cite{patat2020please} identify RTD vulnerabilities for U2F devices by popular websites and propose replacing cookies with a “soft U2F token." The token resists eavesdropping, but not theft by a compromised client.

\section{Conclusion}

FIDO2 has primarily focused on defending against attacks from afar by remote attackers that compromise a password or attempt to phish the user. In this paper, we explored threats from local attacks on FIDO2 that have received less attention---a browser extension compromise and attackers gaining physical access to an HSK.

We systematically analyzed local attacks on FIDO2. We demonstrated that some attacks are feasible against popular websites supporting FIDO2. Our systematization reveals four underlying flaws that lead to these attacks.

In addition, we showed that the threat to FIDO2 authentication from compromised browser extensions is real by providing data showing that many extensions have sufficient permissions to attack FIDO2. We then conducted a static and dynamic analysis of existing extensions and found no evidence that these attacks occur in the wild.

We conducted two user studies and found that participants did not detect the attacks from current error messages, email alerts, and other UX responses to the attacks.
Future work can explore ways to effectively alert users when these attacks occur.

The move to passwordless authentication may cause attackers to pursue new attack vectors, so these results can play an important role in understanding new potential weaknesses that attackers can exploit.

\paragraph{\textbf{Disclosure}}
We shared our research paper with Yubico and Firefox. Although the attacks are outside the current threat model, they appreciated receiving our results for future consideration.

%This paper provides a comprehensive examination of attacks against FIDO2 implementations by a compromised browser extension. 
%We identified seven attacks that a compromised browser extension can lunch. In addition, we demonstrated the attacks work on ten popular websites supporting FIDO2. 
%We provided recommendations to RPs on how to avoid some of the attacks.
%We also identified a stealthy device cloning attack that temporarily avoids the current device-cloning detection mechanism in FIDO2. We then proposed a new, backward-compatible algorithm that detects the stealthy attack using a hashed challenge. 

%We also explore two WebAuthn API design improvements.
%The first proposal modifies the browser architecture to prevent MITM attacks on FIDO2. 
%The second proposal is a design for v-FIDO2, a security model for FIDO2 that leverages a verifier \ver running in secure hardware.  
%We provided a threat analysis describing how well the v-FIDO2 design mitigates attacks on FIDO2 {\hsk}s.
%We also built a proof-of-concept prototype of v-FIDO2 using Intel SGX.
%The first improvement defends against some of the attacks and is easy to deploy. The latter defends against all attacks but is more challenging to deploy. 

%-------------------------------------------------------------------------------
\paragraph{\textbf{Availability}}
%-------------------------------------------------------------------------------

We plan to release the code for FIDO2 attack demonstration to researchers upon request.

\bibliographystyle{IEEEtransS}
{
\bibliography{References}
}
% \pagebreak

\appendix
\section{Appendix}

\begin{figure*}
    \begin{subfigure}{0.30\textwidth}
        \centering
        \includegraphics[ width=\linewidth]{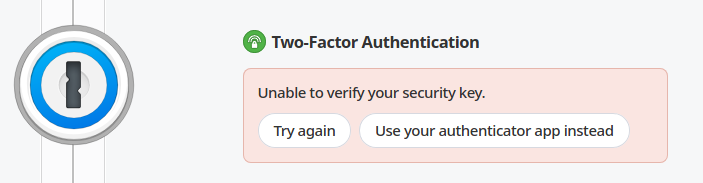}  
        \caption{1Password}
        \label{fig:sub-first1}
    \end{subfigure} %
    \begin{subfigure}{0.30\textwidth}
        \centering
        \includegraphics[trim={0 0 0 2.2cm},clip,width=\linewidth]{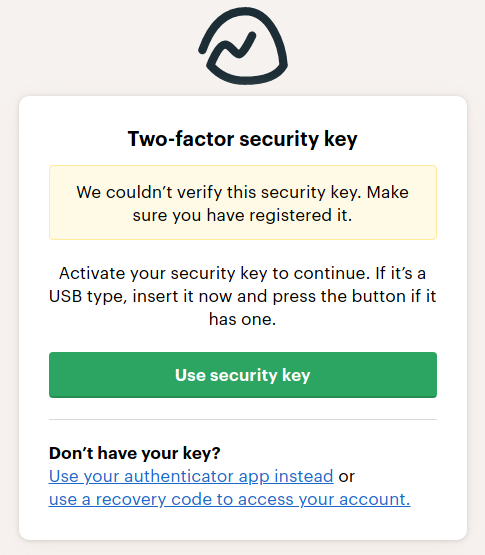} 
        \caption{Basecamp}
        \label{fig:sub-first2}
    \end{subfigure}
    \begin{subfigure}{0.30\textwidth}
        \centering
        \includegraphics[width=\linewidth]{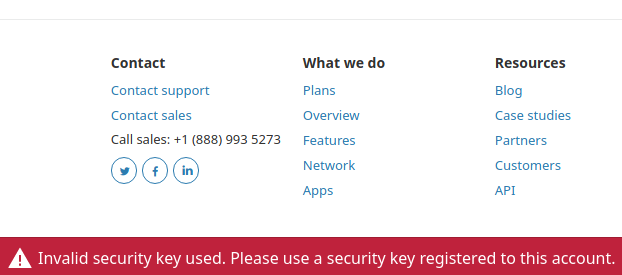}  
        \caption{Cloudflare}
        \label{fig:sub-first3}
    \end{subfigure} %
    
    \begin{subfigure}{0.30\textwidth}
        \centering
        \includegraphics[trim={0 3cm 0 3cm},clip,width=\linewidth]{images/clone_detection_errors/github.png}  
        \caption{Github}
        \label{fig:sub-first4}
    \end{subfigure} %
    \begin{subfigure}{0.30\textwidth}
        \centering
        \includegraphics[trim={7cm 7cm 7cm 3cm},clip,width=\linewidth]{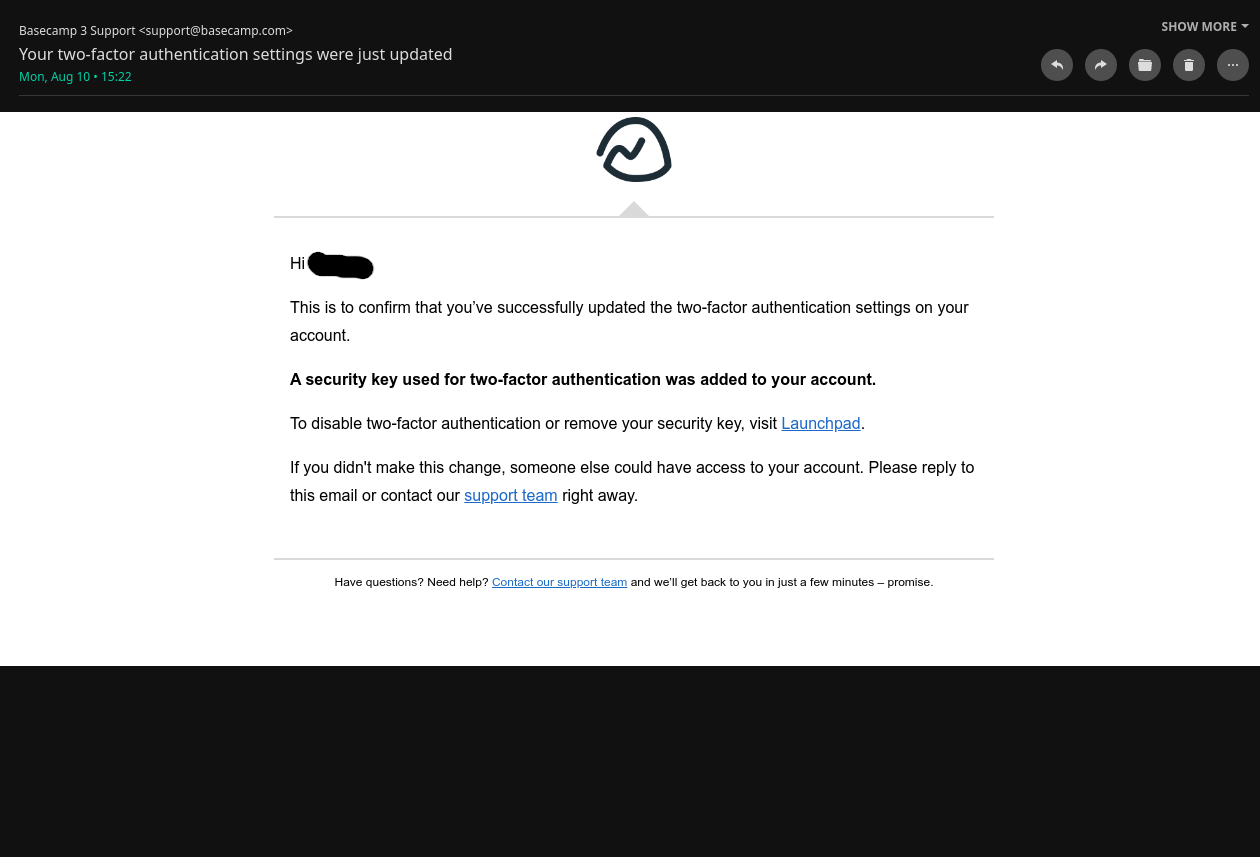}
        \caption{Basecamp}
        \label{fig:sub-first5}
    \end{subfigure}
    \begin{subfigure}{0.30\textwidth}
        \centering
        \includegraphics[trim={0 0 4cm 16.5cm},clip, width=\linewidth]{images/email_notification/github.png}  
        \caption{Github}
        \label{fig:sub-first6}
    \end{subfigure} %

    %\caption{(a-d) Clone detection errors, (e-f) Email notifications on adding a new \hsk}
    \caption{Examples of clone detection error messages (a-d) and email notifications  when adding a new \hsk (e-f)}
    \label{fig:user_confusion_errors_confusion}
\end{figure*}

\begin{algorithm}
\begin{pseudo}
% onReceivingRegistrationRequest(req)\\ 
% Input: req\\
updateChallengeHashHSK (req)\\+
    if $req.reg$\\+
    $hashC$ <- $hash(req.Challenge)$ \\-
    else if $req.auth$\\+
    $resp$ <- $genAuthnResp(hashC)$ \\
    $hashC$ <- $hash(req.Challenge)$
\end{pseudo}
   \caption{Update Challenge hash on the HSK}
   \label{alg:clone-hsk}
\end{algorithm}

\begin{algorithm}
\begin{pseudo}
insertCurrentChallengeHashRP(req) \\+
    $hashList.add(hash(req.challenge))$ \\-
    \\

updateChallengeHashRP(req, resp) \\+
% onReceivingAuthenticationResponse(request, response) \\+
    if $resp.reg$\\+
    $hashList$ <- $list()$ \\
    $hashList.add(hash(req.Challenge))$ \\-
    else if $resp.auth$\\+
    $hashC$ <- $resp.hashC$ \\
    if $!(hashList.contains(hashC))$\\+
        $cloning$ $detected$ \\-
    else \\+
        repeat \\+
            $head$ <- $hashList.remove()$ \\-
        until $head$ == $hashC$ \\-
\end{pseudo}
\caption{Update Challenge hash on the RP}
\label{alg:clone-RP}
\end{algorithm}

% \begin{table*}
% \caption{Themes for user's misunderstanding of clone detection error message.}
% \label{table:user_study_themes}
% \begin{tabular}{|p{8cm}|p{9cm}|}
% User's reasoning for the clone detection errors shown to them & Next step if they want to login to the account after the error message\\
% \hline

% The key is not plugged in correctly. & Unplug, then plug it in again\\
% Wrong hardware token. & Refresh and try again.\\
% The computer not reading the USB correctly. & Use an alternative 2FA method, such as an authentication app, recovery code\\
% Dust on the hardware token or a USB slot & Reset the security key PIN and try again.\\
% The security key is corrupted. & clean USB port\\
% The security key's PIN might be corrupted or outdated. & re-enter username and password\\
% The device could have drifted out of sync and needs to be reset/recalibrated to sync properly. & Restart my computer.\\
% Wrong username or password. & look up help on google\\
% Due to not interacting with the prompt quickly enough & Email GitHub help\\
% Security token expired & \\
% Computer not updated & \\
% Not an error, but it is another layer of authentication & \\
% Server error & \\

% \end{tabular}
% \label{table:Pairwise_comparisons_from_one-way_ANOVA_on_each_relationship}
% \end{table*}
% \input{algorithms}

\end{document}